\documentclass[conference]{IEEEtran}

\usepackage{amsthm}
\usepackage{amsmath}
\usepackage{amssymb}
\usepackage{float}
\usepackage{graphics}
\usepackage{graphicx}
\usepackage[skip=0pt]{caption}
\usepackage{algorithm}
\usepackage{algpseudocode}
\usepackage{bm}
\usepackage{color}
\usepackage{multirow}
\usepackage{makecell}
\usepackage{balance}
\usepackage[captionskip=0pt]{subfig}

\allowdisplaybreaks

\newtheorem{assumption}{Assumption}
\newtheorem{thm}{Theorem}
\newtheorem{lem}{Lemma}

\newcommand{\myparatight}[1]{\smallskip\noindent{\bf {#1}:}~}
\newcommand{\xc}[1]{{{#1}}}
\newcommand{\RomanNumeralCaps}[1]
    {\MakeUppercase{\romannumeral #1}}
\algnewcommand\algorithmicforpara{\textbf{for}}
\algnewcommand\algorithmicdoinparallel{\textbf{do in parallel}}
\algdef{S}[FOR]{ForParallel}[1]{\algorithmicforpara\ #1\ \algorithmicdoinparallel}
\DeclareMathOperator*{\argmin}{arg\,min}
\begin{document}
%
\title{FLTrust: Byzantine-robust Federated Learning via Trust Bootstrapping}

\author{%
  \IEEEauthorblockN{%
    Xiaoyu Cao$^{*1}$\thanks{$^*$Equal contribution.},
    Minghong Fang$^{*2}$,
    Jia Liu$^2$,
    Neil Zhenqiang Gong$^1$%
  }%
  \IEEEauthorblockA{$^1$ Duke University, \{xiaoyu.cao, neil.gong\}@duke.edu}%
  \IEEEauthorblockA{$^2$ The Ohio State University, \{fang.841, liu.1736\}@osu.edu}%
}


%


\IEEEoverridecommandlockouts
\makeatletter\def\@IEEEpubidpullup{6.5\baselineskip}\makeatother
\IEEEpubid{\parbox{\columnwidth}{
    Network and Distributed Systems Security (NDSS) Symposium 2021\\
    21-24 February 2021, San Diego, CA, USA\\
    ISBN 1-891562-66-5\\
    https://dx.doi.org/10.14722/ndss.2021.24434\\
    www.ndss-symposium.org
}
\hspace{\columnsep}\makebox[\columnwidth]{}}

\maketitle


\begin{abstract}
Byzantine-robust federated learning aims to enable a service provider to learn an accurate global model when a bounded number of clients are malicious. The key idea of existing Byzantine-robust federated learning methods is that the service provider performs statistical analysis among the  clients' local model updates and removes suspicious ones, before aggregating them to update the global model. However, malicious clients can still corrupt the global models in these methods via sending carefully crafted local model updates to the service provider. The fundamental reason is that there is no root of trust in existing federated learning methods, i.e., from the service provider's perspective, every client could be malicious. 

In this work, we bridge the gap via proposing \emph{FLTrust}, a new federated learning method in which the service provider itself bootstraps trust. In particular, the service provider itself collects a clean small training dataset (called \emph{root dataset}) 
for the learning task and the service provider maintains a model (called \emph{server model}) based on it to bootstrap trust. In each iteration, the service provider first assigns a trust score to each local model update from the clients,  where a local model update has a lower trust score if its direction deviates more from the direction of the server model update. Then, the service provider normalizes the magnitudes of the local model updates such that they lie in the same hyper-sphere as the server model update in the vector space. Our normalization limits the impact of malicious local model updates with large magnitudes. Finally, the service provider computes the average of the normalized local model updates weighted by their trust scores as a global model update, which is used to update the global model.  Our extensive evaluations on six datasets from different domains show that our FLTrust is secure against both existing attacks and strong adaptive attacks. For instance, using a root dataset with less than 100 examples, FLTrust under adaptive attacks with 40\%-60\% of malicious clients can still train global models that are as accurate as the global models trained by FedAvg under no attacks, where FedAvg is a popular method in non-adversarial settings. 
\end{abstract}


\section{Introduction}
Federated learning (FL) \cite{Konen16,McMahan17} is an emerging distributed learning paradigm on decentralized data. In FL, there are multiple clients (e.g., smartphones, IoT devices, and edge devices) and a service provider (e.g., Google, Apple, and IBM). Each client holds a local training dataset; and the service provider enables the clients to jointly learn a model (called \emph{global model}) without sharing their raw local training data with the service provider. Due to its potential promise of protecting private/proprietary client data, particularly in the age of emerging privacy regulations such as General Data Protection Regulation (GDPR), FL has been deployed by high-profile companies. For instance, Google has deployed FL for next-word prediction on Android Gboard~\cite{gboard}; WeBank uses FL for credit risk prediction~\cite{webank}; and more than 10 leading pharmaceutical companies leverage FL for drug discovery in the project MELLODDY~\cite{melloddy}. Roughly speaking, FL iteratively performs the following three steps: 
the server provided by the service provider sends the current {global model} to the clients or a selected subset of them; each selected client trains a model (called \emph{local model}) via fine-tuning the global model using its own local training data and sends the local model updates back to the server\footnote{It is algorithmically equivalent to send local models instead of their updates to the server.}; and the server aggregates the local model updates to be a \emph{global model update} according to an \emph{aggregation rule} and uses it to update the global model. For instance, FedAvg~\cite{McMahan17}, a popular FL method  in non-adversarial settings developed by Google, computes the average of the local model updates weighted by the sizes of local training datasets as the global model update.   

However, due to its distributed nature, FL is vulnerable to  adversarial manipulations on malicious clients, which could be fake clients injected by an attacker or genuine clients compromised by an attacker. For instance, malicious clients can corrupt the global model via  poisoning their local training data (known as \emph{data poisoning attacks}~\cite{biggio2012poisoning,Nelson08poisoningattackSpamfilter}) or their local model updates sent to the server (called \emph{local model poisoning attacks}~\cite{fang2019local,bhagoji2019analyzing,bagdasaryan2020backdoor,xie2019dba}). The corrupted global model makes incorrect predictions for a large number of testing examples indiscriminately (called \emph{untargeted attacks}) \cite{fang2019local}, or it predicts attacker-chosen target labels for  attacker-chosen target testing examples while the predicted labels for other non-target testing examples are unaffected (called \emph{targeted attacks}) \cite{bagdasaryan2020backdoor,bhagoji2019analyzing,xie2019dba}. For instance, the global model in FedAvg can be arbitrarily manipulated by a single malicious client~\cite{Blanchard17,Yin18}.

Byzantine-robust FL methods \cite{Blanchard17,ChenPOMACS17,Mhamdi18,yang2019byzantine,Yin18} 
aim to address malicious clients. The goal therein is to learn an accurate global model when a bounded number of clients are malicious. Their key idea is to 
leverage \emph{Byzantine-robust aggregation rules}, which essentially compare the clients' local model updates and remove statistical outliers before using them to update the global model. For instance, Median~\cite{Yin18} computes the coordinate-wise median of the clients' local model updates as the global model update. However, recent studies~\cite{bhagoji2019analyzing,fang2019local} showed that existing Byzantine-robust FL methods are still vulnerable to local model poisoning attacks on malicious clients. The fundamental reason is that they have no root of trust. Specifically, from the server's perspective, every client could be malicious, providing no root of trust for the server to decide which local model updates are suspicious.

\myparatight{Our work} In this work, we propose a new Byzantine-robust FL method called \emph{FLTrust}. Instead of completely relying on the local model updates from clients, the server itself bootstraps trust in FLTrust. Specifically, the service provider manually collects a small clean training dataset (called \emph{root dataset}) for the learning task.
The server maintains a model (called \emph{server model}) for the root dataset just like how a client maintains a local model. In each iteration, the server updates the global model by considering both its server model update and the clients' local model updates. 

\myparatight{Our new Byzantine-robust aggregation rule} Specifically, we design a new Byzantine-robust aggregation rule in FLTrust to incorporate the root of trust.  
A model update can be viewed as a vector, which is characterized by its \emph{direction} and \emph{magnitude}. An attacker can manipulate both the {directions} and {magnitudes} of  the local model updates on the malicious clients. Therefore, our aggregation rule takes  both the directions and magnitudes into considerations when computing the global model update. Specifically, the server first assigns a trust score (TS) to a local model update, where the trust score is larger if the direction of the local model update is more similar  to that of the server model update. Formally, we use the \emph{cosine similarity} between a local model update and the server model update to measure the similarity of their directions. However, the cosine similarity alone is insufficient because a local model update, whose cosine similarity score is negative, can still have a negative impact on the aggregated global model update. Therefore, we further clip the cosine similarity score using the popular ReLU operation. The ReLU-clipped cosine similarity is our trust score.
 Then, FLTrust  normalizes each local model update by scaling it to have the same magnitude as the  server model update. Such normalization essentially projects each local model update to the same hyper-sphere where the server model update lies in the vector space, which limits the impact of the poisoned local model updates with large magnitudes. Finally, FLTrust computes the average of the normalized local model updates weighted by their trust scores as the global model update, which is used to update the global model.

\myparatight{FLTrust can defend against existing attacks} We perform extensive empirical evaluation on six datasets from different domains, including five image classification datasets (MNIST-0.1, MNIST-0.5, Fashion-MNIST, CIFAR-10, and CH-MNIST) and a smartphone-based human activity recognition dataset (Human Activity Recognition). We compare FLTrust with multiple existing Byzantine-robust FL methods including Krum~\cite{Blanchard17}, Trimmed mean~\cite{Yin18}, and Median~\cite{Yin18}. Moreover, we evaluate multiple poisoning attacks including label flipping attack (a data poisoning attack), Krum attack and Trim attack (untargeted local model poisoning attacks)~\cite{fang2019local}, as well as Scaling attack\footnote{The Scaling attack is also known as a backdoor attack.} (targeted local model poisoning attack)~\cite{bagdasaryan2020backdoor}. Our results show that FLTrust is secure against these attacks even if the root dataset has less than 100 training examples, while existing Byzantine-robust FL methods are vulnerable to them or a subset of them. For instance, a CNN global model learnt using FLTrust has a testing error rate of 0.04 under all the evaluated attacks on MNIST-0.1. However, the Krum attack can increase the testing error rate of the CNN global model learnt by Krum from 0.10 to 0.90. Moreover, we treat FedAvg under no attacks as a baseline and compare our FLTrust under attacks with it. Our results show that FLTrust under attacks achieves similar testing error rates to FedAvg under no attacks. We also study different variants of FLTrust and the impact of different system parameters on FLTrust. For instance, our results show that FLTrust works well once the root dataset distribution does not diverge too much from the overall training data distribution of the learning task. 

\myparatight{FLTrust can defend against adaptive attacks} 
An attacker can adapt its attack to  FLTrust. Therefore, we also evaluate FLTrust against adaptive attacks. Specifically, Fang et al.~\cite{fang2019local} proposed a general framework of local model poisoning attacks, which can be applied to optimize the attacks for any given aggregation rule. An attacker can substitute the aggregation rule of FLTrust into the framework and obtain an adaptive attack that is particularly optimized against FLTrust. Our empirical results show that FLTrust is still robust against such adaptive attacks. For instance, even when 60\% of the clients are malicious and collude with each other, FLTrust can still learn a CNN global model with testing error rate 0.04 for  MNIST-0.1. This testing error rate is the same as that of the CNN global model learnt by FedAvg under no attacks. 

Our contributions can be summarized as follows: 

\begin{itemize}
    \item We propose the first federated learning method FLTrust that bootstraps trust to achieve Byzantine robustness against malicious clients.  
    \item We empirically evaluate FLTrust against existing attacks. Our results show that FLTrust can defend against them. 
    \item We design adaptive attacks against FLTrust and evaluate their performance. Our results show that FLTrust is also robust against the adaptive attacks. 
\end{itemize}


\section{Background and Related Work}

\subsection{Background on Federated Learning (FL)} 
\label{sec:background}

\begin{figure*}[!t]
	\centering
	{\includegraphics[width= 0.8\textwidth]{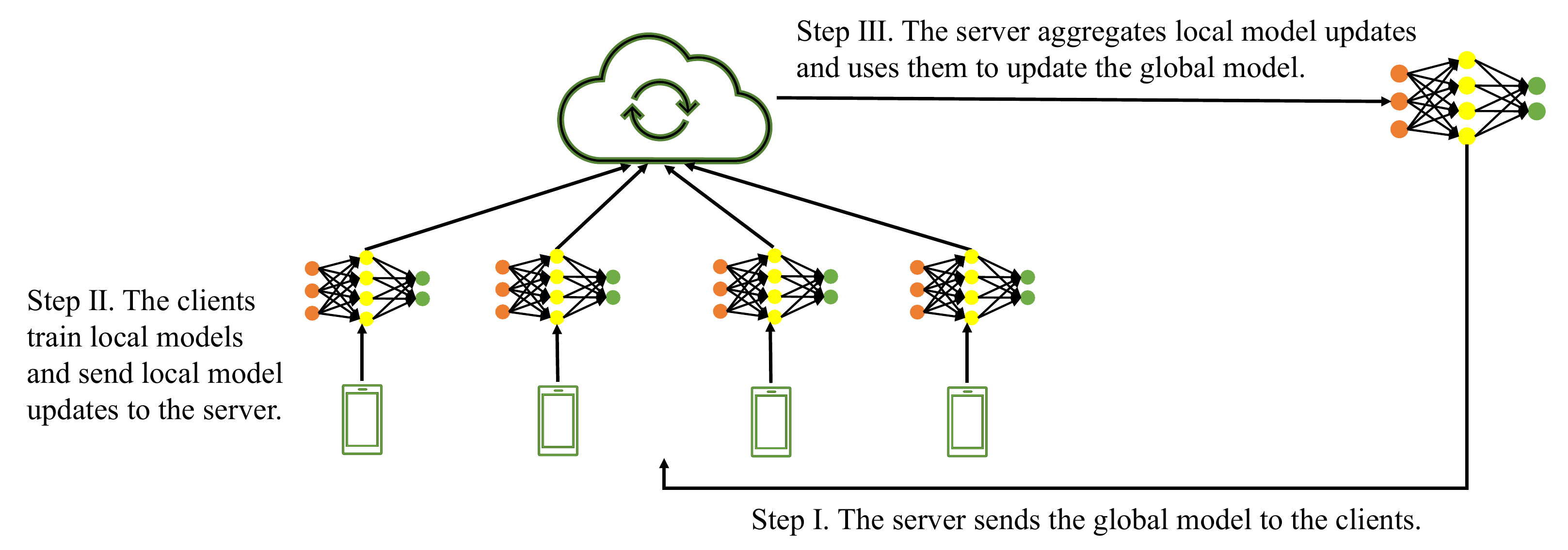}}
	\vspace{2mm}
	\caption{Illustration of the three steps in FL.}
	\label{fig:fl}
	\vspace{-4mm}
\end{figure*}

\xc{Suppose we have $n$ clients and each client  has a local training dataset $D_i, i=1,2,\cdots,n$. We use $D=\bigcup_{i=1}^{n} D_{i}$ to denote the joint training data. Each training example in $D$ is drawn from an unknown distribution $\mathcal{X}$. The clients aim to collaboratively learn a shared {global model} with the help of a service provider. 
 The optimal global model $\bm{w}^*$ is a solution to the following optimization problem: $\bm{w}^* = \argmin_{\bm{w}} F(\bm{w})$, where $F(\bm{w})=\mathbb{E}_{D\sim\mathcal{X}}[f(D,\bm{w})]$ is the expectation of the empirical loss $f(D,\bm{w})$ on the joint training dataset $D$. Since the expectation is hard to evaluate, the global model is often learnt via minimizing the empirical loss in practice, i.e., $\argmin_{\bm{w}} f(D,\bm{w})$ is the learnt global model.} Specifically, each client  maintains a local model for its local training dataset. Moreover,  a service provider's server maintains the global model via aggregating the local model updates from the clients. Specifically, FL iteratively performs the following three steps (illustrated in Figure \ref{fig:fl}):
\begin{itemize}
    \item \textbf{Step I: Synchronizing the global model with clients.} The server sends the current global model $\bm{w}$ to the clients or a subset of them. 
    \item \textbf{Step II: Training local models.} Each client trains a local model via fine-tuning the global model using  its local training dataset. Formally, the $i$th client solves the optimization problem  $\min_{\bm{w}_i} f(D_i,\bm{w}_i)$, where $\bm{w}_i$ is the client's local model. In particular, the client initializes its local model as the global model and uses stochastic gradient descent to update the local model for one or more iterations. Then, each client sends its local model update $\bm{g}_i=\bm{w}_i - \bm{w}$ (i.e., the difference between its local model and the current global model) to the server. 
    \item \textbf{Step III: Updating the global model via aggregating the local model updates.}  The server computes a \emph{global model update} $\bm{g}$ via aggregating the local model updates according to some \emph{aggregation rule}. Then, the server updates the global model using the global model update, i.e., $\bm{w} = \bm{w} + \alpha \cdot \bm{g}$, where $\alpha$ is the global learning rate.
\end{itemize}

The aggregation rule plays a key role in FL. Different FL methods essentially use different aggregation rules. Next, we discuss popular aggregation rules.  

\subsubsection{FedAvg}
FedAvg \cite{McMahan17} was proposed by Google. FedAvg computes the average of the clients' local model updates as the global model update, where each client is weighted by its number of training examples. Formally, $\bm{g}=\sum_{i=1}^n \frac{|D_i|}{N} \bm{g}_i$, where $|D_i|$ is the local training dataset size on the $i$th client and $N$ is the total number of training examples. FedAvg is the state-of-the-art FL method in non-adversarial settings. However, the global model in FedAvg can be arbitrarily manipulated by a single malicious client \cite{Blanchard17,Yin18}.  

\subsubsection{Byzantine-robust Aggregation Rules} 
\xc{Most Byzantine-robust FL methods use Byzantine-robust aggregation rules (see, e.g., \cite{Blanchard17,ChenPOMACS17,Mhamdi18,rajput2019detox,xie2019zeno,yang2019byrdie,Yin18,munoz2019byzantine}) that aim to tolerate Byzantine client failures. One exception is that Li et al.~\cite{li2019rsa} introduced a norm regularization term into the loss function.} Examples of Byzantine-robust aggregation rules include Krum \cite{Blanchard17}, Trimmed mean \cite{Yin18}, and Median \cite{Yin18}, which we discuss next. 

\myparatight{Krum \cite{Blanchard17}}
Krum selects one of the $n$ local model updates in each iteration as the global model update based on a square-distance score. Suppose at most $f$ clients are malicious. The score for the $i$th client is computed as follows:
\begin{align}
    s_i = \sum_{\bm{g}_j\in\Gamma_{i,n-f-2}} \Vert\bm{g}_j-\bm{g}_i\Vert_2^2,
\end{align}
where $\Gamma_{i,n-f-2}$ is the set of $n-f-2$ local model updates that have the smallest Euclidean distance to $\bm{g}_i$. The local model update of the client with the minimal score will be chosen as the global model update to update the global model. 

\myparatight{Trimmed Mean (Trim-mean) \cite{Yin18}} Trimmed mean is a coordinate-wise aggregation rule that considers each model parameter individually. For each model parameter, the server collects its values in all local model updates and sorts them. Given a trim parameter $k<\frac{n}{2}$, the server removes the largest $k$ and the smallest $k$ values, and then computes the mean of the remaining $n-2k$ values as the value of the corresponding parameter in the global model update. The trim parameter $k$ should be at least the number of malicious clients to make Trim-mean robust. In other words, Trim-mean can tolerate less than 50\% of malicious clients.   

\myparatight{Median \cite{Yin18}} Median is another coordinate-wise aggregation rule. Like Trim-mean, in Median, the server  also sorts the values of each individual parameter in all local model updates. Instead of using the mean value after trim, Median considers the median value of each parameter as the corresponding parameter value in the global model update. 

Existing FL methods suffer from a key limitation: they are vulnerable to sophisticated local model poisoning attacks on malicious clients, which we discuss in the next section. 

\subsection{Poisoning Attacks to Federated Learning}

Poisoning attacks generally refer to attacking the training phase of machine learning. 
One category of poisoning attacks called \emph{data poisoning attacks} aim to pollute the training data to corrupt the learnt model. Data poisoning attacks have been demonstrated to many machine learning systems such as spam detection~\cite{Nelson08poisoningattackSpamfilter,rubinstein2009antidote}, SVM~\cite{biggio2012poisoning}, recommender systems~\cite{fang2020influence,fang2018poisoning,poisoningattackRecSys16,YangRecSys17}, neural networks~\cite{Chen17,Gu17,liu2017trojaning,munoz2017towards,shafahi2018poison,Suciu18}, and graph-based methods~\cite{jia2020certified,Wang19,zhang2020backdoor}, as well as  distributed privacy-preserving data analytics~\cite{cao2019data,cheu2019manipulation}. 
FL is also vulnerable to data poisoning attacks~\cite{tolpegin2020data}, i.e., malicious clients can corrupt the global model via modifying, adding, and/or deleting examples in their local training datasets. For instance, a data poisoning attack known as \emph{label flipping attack} changes the labels of the training examples on malicious clients while keeping their features unchanged. 

Moreover, unlike centralized learning, FL is further vulnerable to \emph{local model poisoning attacks}~\cite{bagdasaryan2020backdoor,bhagoji2019analyzing,fang2019local,xie2019dba}, in which the malicious clients poison the local models or their updates sent to the server. Depending on the attacker's goal, local model poisoning attacks can be categorized into \emph{untargeted attacks} \cite{fang2019local} and \emph{targeted attacks} \cite{bagdasaryan2020backdoor,bhagoji2019analyzing,xie2019dba}.  Untargeted attacks aim to   corrupt the global model such that it makes incorrect predictions for a large number of testing examples indiscriminately, i.e., the testing error rate is high. Targeted attacks aim to  corrupt the global model such that it predicts attacker-chosen target labels for attacker-chosen target testing examples while the predicted labels for other non-target testing examples are unaffected. 

Note that any data poisoning attack can be transformed to a local model poisoning attack, i.e., we can compute the local model update on a malicious client's poisoned local training dataset and treat it as the poisoned local model update. Moreover, recent studies \cite{bhagoji2019analyzing,fang2019local} showed that local model poisoning attacks are more effective than data poisoning attacks against FL. Therefore, we focus on local model poisoning attacks in this work. Next, we discuss two state-of-the-art untargeted attacks (i.e., Krum attack and Trim attack) \cite{fang2019local} and one targeted attack (i.e., Scaling attack) \cite{bagdasaryan2020backdoor}. 

\myparatight{Krum attack and Trim attack \cite{fang2019local}} Fang et al.~\cite{fang2019local}  proposed a general framework for local model poisoning attacks, which can be applied to optimize the attacks for any given aggregation rule. 
Assuming the global model update without attack is $\bm{g}$,  
Fang et al.~\cite{fang2019local} formulate the attack as an optimization problem  that aims to change the global model update the most along the opposite direction of $\bm{g}$, by optimizing the poisoned local model updates sent from the malicious clients to the server. 
 Different aggregation rules lead to different instantiations of the optimization problem. 
Fang et al. applied the framework to optimize local model poisoning attacks for Krum (called Krum attack) as well as Trim-mean and Median (called Trim attack). 

\myparatight{Scaling attack \cite{bagdasaryan2020backdoor}} This attack aims to corrupt the global model to predict attacker-chosen target labels for attacker-chosen target testing examples, while the predicted labels for other testing examples are unaffected (i.e., the normal testing error rate remains the same). For instance, the attacker-chosen target testing examples can be normal testing examples embedded with a predefined backdoor trigger (e.g., a logo, a specific feature pattern). To achieve the goal, the Scaling attack adds trigger-embedded training examples with the attacker-chosen target label to the local training data of malicious clients. The local model updates on malicious clients are then computed based on the local training datasets augmented with the trigger-embedded examples. However, the poisoned local model updates may have limited impact on the global model update because it is aggregated over all clients' local model updates. For instance, in FedAvg, the effect of the attack will be diluted after the averaging~\cite{bagdasaryan2020backdoor}. Therefore, the attack further scales the poisoned local model updates on malicious clients by a factor that is much larger than 1. The scaled poisoned local model updates are then sent to the server. 

\begin{figure*}[!t]
	\centering
	{\includegraphics[width= 0.8\textwidth]{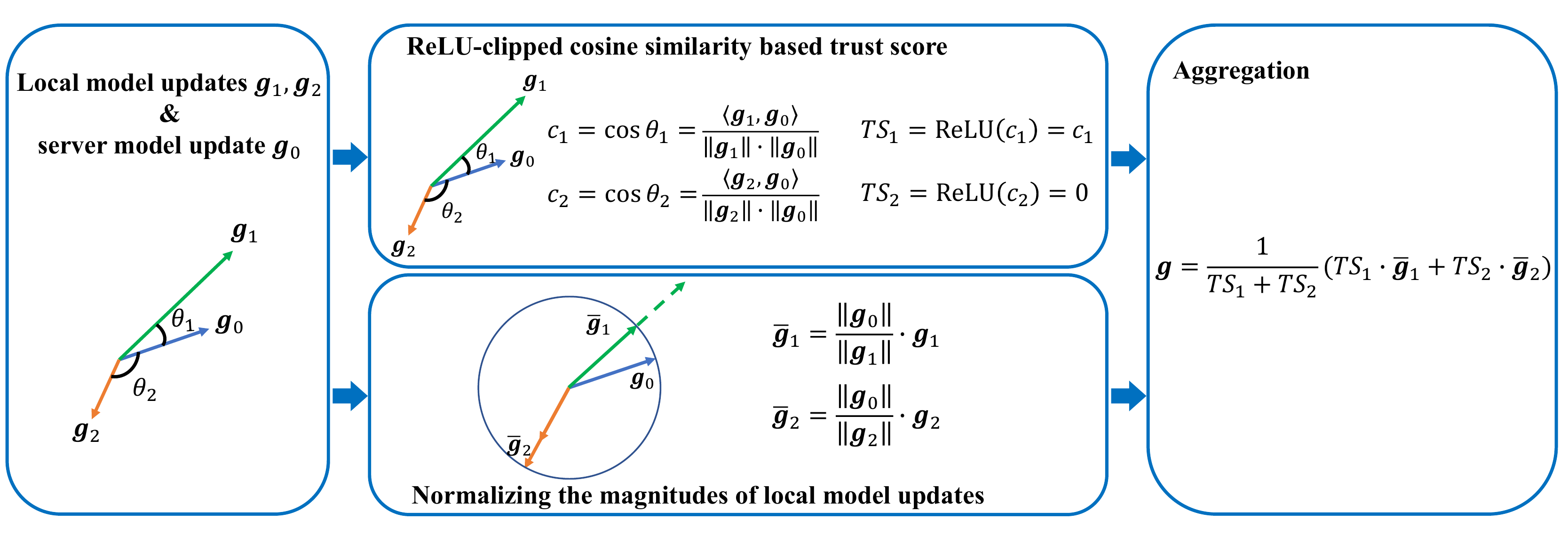}}
	\vspace{2mm}
	\caption{Illustration of our aggregation rule, which is applied in each iteration of FLTrust.} 
	\label{fltrust}
\end{figure*}


\section{Problem Setup}
\label{sec:prob}

\myparatight{Attack model} We follow the attack model in previous works~\cite{bagdasaryan2020backdoor,bhagoji2019analyzing,fang2019local}. Specifically, an attacker controls some malicious clients, which can be fake clients injected by the attacker or genuine ones compromised by the attacker. However, the attacker does not compromise the server. The malicious clients can send arbitrary local model updates to the server in each iteration of  the FL training process. \xc{Typically, an attacker has the following partial knowledge about an FL system: local training data and local model updates on the malicious clients, loss function, and learning rate. We notice that the Scaling attack \cite{bagdasaryan2020backdoor} only requires such partial knowledge. The Krum and Trim attacks \cite{fang2019local} are also applicable in this partial-knowledge setting. However, they are stronger in the full-knowledge setting \cite{fang2019local}, where the attacker knows everything about the FL training process, including the local training data and local model updates on all clients in each iteration, as well as the FL's aggregation rule. Therefore, we consider such full-knowledge setting to show that our method can defend against strong attacks.} Moreover, the attacker can perform adaptive attacks to FLTrust, which we discuss in Section \ref{sec:adaptive}.

\myparatight{Defense goals} We aim to design an FL method that achieves Byzantine robustness against malicious clients without sacrificing the fidelity and efficiency. In particular, we treat FedAvg under no attacks as a baseline to discuss fidelity and efficiency, i.e., our method should be robust against malicious clients while being as accurate and efficient as FedAvg under no attacks. Specifically, we aim to design a Byzantine-robust FL method that achieves the following defense goals:
\begin{itemize}
    \item {\bf Fidelity.} The method should not sacrifice the classification accuracy of the global model when there is no attack. In particular, under no attacks, the method should be able to learn a global model that is as accurate as  the global model learnt by FedAvg, a popular FL method in non-adversarial settings. 
    \item {\bf Robustness.} The method should preserve the classification accuracy of the global model in the presence of  malicious clients performing strong poisoning attacks. In particular, we aim to design a method that can learn a global model under attacks that is as accurate as the global model learnt by FedAvg under no attacks. Moreover, for targeted attacks, our goal further includes that the global model is unlikely to predict the attacker-chosen target labels for the attacker-chosen target testing examples. 
    \item {\bf Efficiency.} The method should not incur extra computation and communications overhead, especially to the clients. Clients in FL are often resource-constrained devices. Therefore, we aim to design a method that does not increase the workload of the clients, compared to FedAvg under no attacks. 
\end{itemize}

Existing Byzantine-robust FL methods such as Krum, Trim-mean, and Median do not satisfy the fidelity and robustness goals. Moreover, Krum does not satisfy the efficiency goal because it requires the server to compute pairwise distances of the clients' local model updates, which is computationally expensive when the number of clients is large. 

\myparatight{Defender's knowledge and capability} We consider the defense is performed on the server side. The server does not have access to the raw local training data on the clients, and the server does not know the number of malicious clients. However, the server has full access to the global model as well as the local model updates from all clients in each iteration.  Moreover, the server itself can collect a clean small training dataset (we call it root dataset) for the learning task. 
\xc{We require the root dataset to be clean from poisoning. The server can collect a clean root dataset by manual labeling. For instance, Google enlists its employees to type with Gboard to create the root dataset for its federated next-word prediction \cite{gboard}; when the learning task is digit recognition, the service provider can hire human workers to label some digits. Since we only require a small root dataset, e.g., 100 training examples, it is often affordable for the server to perform manual collection and labeling. The root dataset may or may not follow the same distribution as the overall training data distribution of the learning task. Our experimental results show that our method is effective once the root dataset distribution does not deviate too much from the overall training data distribution.} 


\section{Our FLTrust}

\subsection{Overview of FLTrust}

In our FLTrust, the server itself collects a small clean training dataset (called \emph{root dataset}) and maintains a model (called \emph{server model}) for it just like how a client maintains a local model. In each iteration, our FLTrust follows the general three steps of FL discussed in Section~\ref{sec:background}. However, our FLTrust is different from existing FL methods in Step II and Step III. Specifically, in Step II, each client trains its local model in existing FL methods, while the server also trains its server model via fine-tuning the current global model using the root dataset in FLTrust. In Step III, existing FL methods only consider the clients' local model updates to update the global model, which provides no root of trust. On the contrary, FLTrust considers both the server model update and the clients' local model updates to update the global model. 

Specifically, an attacker can manipulate the directions of the local model updates on the malicious clients  such that the global model is updated towards the opposite of the direction along which it should be updated; or the attacker  can scale up the magnitudes of the local model updates to dominate the aggregated global model update. Therefore,  we take both the directions and the magnitudes of the model updates into consideration. In particular, FLTrust first assigns a trust score (TS) to a local model update based on its direction similarity with the server model update. Formally, our trust score of a local model update is its ReLU-clipped cosine similarity with the server model update. Then, FLTrust normalizes each local model update by scaling it to have the same magnitude as the  server model update. Such normalization essentially projects each local model update to the same hyper-sphere where the server model update lies in the vector space, which limits the impact of the poisoned local model updates with large magnitudes. Finally, FLTrust computes the average of the normalized local model updates weighted by their trust scores as the global model update, which is used to update the global model. 

\subsection{Our New Aggregation Rule} 

Our new aggregation rule considers both the directions and magnitudes of the local model updates and the server model update to compute the global model update. Figure~\ref{fltrust} illustrates our aggregation rule. 

\myparatight{ReLU-clipped cosine similarity based trust score} An attacker can manipulate the directions of the local model updates on the malicious clients such that the global model update is driven to an arbitrary direction that the attacker desires. Without root of trust, it is challenging for the server to decide which direction is more ``promising'' to update the global model. In our FLTrust, the root trust origins from the direction of the server model update. In particular, if the direction of a local model update is more similar to that of the server model update, then the direction of the local model update may be more ``promising''. Formally, we use the \emph{cosine similarity}, a popular metric to measure the angle between two vectors, to measure the direction similarity between a local model update and the server model update. 

However, the cosine similarity alone faces a challenge. Specifically, if a local model update and the server model update are in  opposite directions, their cosine similarity is negative, which still has negative impact on the aggregated global model update (see our experimental results in Section~\ref{sec:exp_res}). Therefore, we exclude such local model updates from the aggregation by clipping the cosine similarity. In particular, we use the popular ReLU operation for clipping. Formally, our  trust score is  defined as follows:
\begin{align}
    TS_i = ReLU(c_i),
\end{align}
where $TS_i$ is the trust score for the $i$th local model update $\bm{g}_i$, and $c_i$ is the cosine similarity between  $\bm{g}_i$ and the server model update $\bm{g}_0$, i.e., $c_i=\frac{\langle\bm{g}_i,\bm{g}_0 \rangle}{||\bm{g}_i||\cdot ||\bm{g}_0||}$.   ReLU  is defined as follows: $ReLU(x)=x$ if $x>0$ and $ReLU(x)=0$ otherwise. 

\myparatight{Normalizing the magnitudes of local model updates} An attacker can also scale the magnitudes of the  local model updates on the malicious clients by a large factor such that they dominate the global model update. Therefore, we normalize the magnitude of each local model update. Without root of trust, it is challenging to decide what quantity we should normalize to. However,  the server has the root dataset to bootstrap trust in FLTrust. Therefore, we normalize each local model update such that it has the same magnitude as the server model update. Such normalization means that we rescale  local model updates to be the same hyper-sphere where the server model update lies in the vector space. Formally, we have the following: 
\begin{align}
    \bm{\bar{g}}_i =  \frac{||\bm{g}_0||}{||\bm{g}_i||}\cdot \bm{g}_i,   
\end{align}
where $\bm{g}_i$ is the local model update of the $i$th client in the current iteration, $\bm{\bar{g}}_i$ is the \emph{normalized local model update} of the $i$th client, $\bm{g}_0$ is the server model update, and $||\cdot||$ means $\ell_2$ norm of a vector. Our normalization ensures that no single local model update has too much impact on the aggregated global model update. Note that our normalization also enlarges a local model update with a small magnitude to have the same magnitude as the server model update. This is based on the intuition that local model updates with small magnitudes are more likely from benign clients, and thus enlarging their magnitudes helps reduce the impact of the poisoned local model updates from the malicious clients, leading to a better global model (see our experimental results in Section~\ref{sec:exp_res}). 

\begin{algorithm}[t]
	\caption{ModelUpdate($\bm{w}$, $D$, $b$, $\beta$, $R$)}\label{local_upate_algo}
	\begin{algorithmic}[1]
		\renewcommand{\algorithmicensure}{\textbf{Output:}}
		\Ensure Model update. 
		\State $\bm{w}^{0} \leftarrow \bm{w}$.
		\For {$r=1,2,\cdots,R$} 
		    \State Randomly sample a batch $D_b$ from $D$.
		    \State $\bm{w}^r \leftarrow \bm{w}^{r-1} - \beta \nabla Loss(D_b;\bm{w})$.
		\EndFor\\
		\Return $\bm{w}^R - \bm{w}$. 
	\end{algorithmic} 
\end{algorithm}

\begin{algorithm}[t]
	\caption{FLTrust}\label{global_upate_algo}
	\begin{algorithmic}[1]
		\renewcommand{\algorithmicrequire}{\textbf{Input:}}
		\renewcommand{\algorithmicensure}{\textbf{Output:}}
		\Require $n$ clients with local training datasets $D_i, i=1,2,\cdots,n$; a server with root dataset $D_0$; global learning rate $\alpha$; number of global iterations $R_g$; number of clients $\tau$  sampled in each iteration; local learning rate $\beta$; number of local iterations $R_l$; and batch size $b$.
		\Ensure Global model $\bm{w}$. 
		\State $\bm{w} \leftarrow $ random initialization.
		\For{$r=1,2,\cdots,R_g$}
		    \State // Step \RomanNumeralCaps{1}: The server sends the global model to clients.
		    \State The server randomly samples $\tau$ clients $C_1, C_2, \cdots, C_\tau$ from $\{1,2,\cdots,n\}$ and sends $\bm{w}$ to them.
		    \item[]
	        \State // Step \RomanNumeralCaps{2}: Training local models and server model.
	        \State // Client side.
		    \ForParallel{$i=C_1, C_2, \cdots, C_\tau$} 
		        \State $\bm{g}_i = ModelUpdate(\bm{w},D_{i},b,\beta,R_l)$.
		        \State Send $\bm{g}_i$ to the server.
		    \EndFor
		    \State // Server side.
		    \State $\bm{g}_0 = ModelUpdate(\bm{w},D_{0},b,\beta,R_l)$.
		    \item[]
		    \State // Step \RomanNumeralCaps{3}: Updating the global model via aggregating the local model updates.
		    \For{$i=C_1, C_2, \cdots, C_\tau$} 
		        \State $TS_i=ReLU\left(\frac{\langle\bm{g}_i,\bm{g}_0\rangle}{\Vert\bm{g}_i\Vert\Vert\bm{g}_0\Vert}\right)$.
		        \State $\bm{\bar{g}}_i =  \frac{||\bm{g}_0||}{||\bm{g}_i||}\cdot \bm{g}_i$.
		    \EndFor
            \State $\bm{g} = \frac{1}{\sum\limits_{j=1}^{\tau}TS_{C_j}} \sum\limits_{i=1}^\tau {TS}_{C_i} \cdot \bm{\bar{g}}_{C_i}$.
            \State $\bm{w} \leftarrow \bm{w} + \alpha\cdot\bm{g}$.
		\EndFor\\
		\Return $\bm{w}$.
	\end{algorithmic} 
\end{algorithm}

\myparatight{Aggregating the local model updates} We compute the average of the normalized local model updates weighted by their trust scores as the global model update:
\begin{align}
\label{agg_local_model}
    \bm{g} &= \frac{1}{\sum\limits_{j=1}^{n}TS_j} \sum_{i=1}^n {TS}_i \cdot \bm{\bar{g}}_i\nonumber\\
        &= \frac{1}{\sum\limits_{j=1}^{n}ReLU(c_j)} \sum_{i=1}^n {ReLU(c_i)} \cdot \frac{||\bm{g}_0||}{||\bm{g}_i||}\cdot\bm{g}_i,
\end{align}
where $\bm{g}$ is the global model update. Note that if the server selects a subset of clients in an iteration, the global model update is aggregated from the local model updates of the selected clients. In principle, the server model update can be treated as a local model update with a trust score of 1 and the global model update can be weighted average of  the clients' local model updates together with the server model update. However, such variant may negatively impact the global model because the root dataset is small and may not have the same distribution as the training data, but the server model update derived from it has a trust score of 1, reducing the contributions of the benign clients' local model updates (see our experimental results in Section~\ref{sec:exp_res}).  
Finally, we update the global model as follows:
\begin{align}   
\label{update_global_model}
\bm{w} \leftarrow \bm{w} + \alpha \cdot \bm{g},
\end{align}
where $\alpha$ is the global learning rate. 

\subsection{Complete FLTrust Algorithm}
 Algorithm \ref{global_upate_algo} shows our complete FLTrust method.  FLTrust performs $R_g$ iterations and  has three steps in each iteration. In Step I, the server sends the current global model to the clients or a subset of them. In Step II, the clients compute the local model updates based on the global model and their local training data, which are then sent to the server. Meanwhile, the server itself computes the server model update based on the global model and the root dataset. The local model updates and the server model update are computed by the function ModelUpdate in Algorithm~\ref{local_upate_algo} via performing stochastic gradient descent for $R_l$ iterations with a local learning rate $\beta$. In Step III, the server computes the global model update by aggregating the local model updates and uses it to update the global model with a global learning rate $\alpha$.

\xc{
\subsection{Formal Security Analysis}
\label{sec:sec_ana}

As we discussed in Section~\ref{sec:background}, the optimal global model $\bm{w}^*$ is a solution to the following optimization problem: $\bm{w}^* = \argmin_{\bm{w} \in \Theta} F(\bm{w}) \triangleq \mathbb{E}_{D\sim\mathcal{X}} \left[ f(D, \bm{w}) \right]$, where $\Theta$ is the parameter space of the global model, $D=\bigcup_{i=1}^{n} D_{i}$ is the joint training dataset of the $n$ clients, $\mathcal{X}$ is the training data distribution,  $f(D, \bm{w})$ is the empirical loss function on the training data $D$, and $F(\bm{w})$ is the expected loss function.    Our FLTrust is  an iterative algorithm to find a global model to minimize the empirical loss function $f(D, \bm{w})$. We show that, under some assumptions, the difference between the global model learnt by FLTrust under attacks and the optimal global model $\bm{w}^*$ is bounded. Next, we first describe our assumptions and then describe our theoretical results.  

\begin{assumption}
	\label{assumption_1}
	The expected loss function $F(\bm{w})$ is $\mu$-strongly convex and differentiable over the space $\Theta$ with $L$-Lipschitz continuous gradient. Formally, we have the following for any $\bm{w}, \widehat{\bm{w}} \in \Theta$:
	\begin{gather*}
	F(\widehat{\bm{w}}) \ge F(\bm{w}) + \left\langle {\nabla F(\bm{w}),\widehat{\bm{w}} - \bm{w}} \right\rangle + \frac{\mu}{2}{\left\| \widehat{\bm{w}} - \bm{w} \right\|^2},\\
	\left\| {\nabla F(\bm{w}) - \nabla F(\widehat{\bm{w}}) } \right\| \le L\left\| {\bm{w} - \widehat{\bm{w}}} \right\|, 
	\end{gather*}
	where $\nabla$ represents gradient, $\left\|\cdot \right\|$ represents $\ell_2$ norm, and $\left\langle \cdot, \cdot \right\rangle$ represents inner product of two vectors. Moreover, the empirical loss function $f(D,\bm{w})$ is $L_1$-Lipschitz probabilistically. Formally, 	for any $\delta  \in (0,1)$, there exists an $L_1$ such that:
	\begin{align}
	\text{Pr}\left\{ {\mathop {\sup }\limits_{\bm{w}, \widehat{\bm{w}} \in \Theta :\bm{w} \ne \widehat{\bm{w}}} \frac{\left\| {\nabla f(D,\bm{w}) - \nabla f(D,\widehat{\bm{w}})} \right\|} 
		{\left\| {\bm{w} - \widehat{\bm{w}}} \right\|} \le {L_1}} \right\} \ge 1 - \frac{\delta }{3}. \nonumber
	\end{align}
\end{assumption}	

\begin{assumption}
	\label{assumption_2}
	The gradient of the empirical loss function $\nabla f(D,\bm{w}^*)$ at the optimal global model $\bm{w}^*$ is bounded. Moreover, the gradient difference $h(D,\bm{w}) = \nabla f(D,\bm{w}) - \nabla f(D,\bm{w}^*)$ for any $\bm{w} \in \Theta$ is bounded. Specifically, there exist positive constants $\sigma_1$ and $\gamma_1$ such that for any unit vector $\bm{v}$, $\left\langle {\nabla f(D,\bm{w}^*),\bm{v}} \right\rangle$ is sub-exponential with $\sigma_1$ and $\gamma_1$; and there exist positive constants $\sigma_2$ and $\gamma_2$ such that for any $\bm{w} \in \Theta$ with $\bm{w} \ne \bm{w}^*$ and any unit vector $\bm{v}$, $\left\langle {h(D,\bm{w}) - \mathbb{E}\left[ {h(D,\bm{w})} \right],\bm{v}} \right\rangle / \left\| {\bm{w} - \bm{w}^*} \right\|$ is sub-exponential with $\sigma_2$ and $\gamma_2$.  
	Formally, for $\forall \left| \xi  \right| \le 1/\gamma_1$, $\forall \left| \xi  \right| \le 1/\gamma_2$, we have: 
	\begin{equation*}
	\mathop {\sup }\limits_{\bm{v} \in \bm{B}} \mathbb{E}\left[ {\exp (\xi \left\langle {\nabla f(D,\bm{w}^*), \bm{v}} \right\rangle)} \right]
	\le 
	e ^{\sigma_1 ^2 \xi^2 /2},
	\end{equation*}
	\begin{equation*}
	\mathop {\sup }\limits_{\bm{w} \in \Theta ,\bm{v} \in \bm{B}} \mathbb{E} \left[ {\exp \left( {\frac{{\xi \left\langle {h(D,\bm{w}) - \mathbb{E} \left[ {h(D,\bm{w})} \right],\bm{v}} \right\rangle }} {\left\| {\bm{w} - \bm{w}^*} \right\|}} \right)} \right]
	\le 
	e ^{\sigma_2 ^2 \xi^2 /2} ,
	\end{equation*}
	where $\bm{B}$ is the unit sphere $\bm{B}=\left\{ {\bm{v}:{{\left\| \bm{v} \right\|}} = 1} \right\}$.
\end{assumption}	

\begin{assumption}
	\label{assumption_3}
	Each client's local training dataset $D_i$ ($i=1,2,\cdots,n$) and the root dataset $D_0$ are sampled independently from the distribution $\mathcal{X}$.
\end{assumption}

\begin{thm}
	\label{theorem_1}
	Suppose Assumption~\ref{assumption_1}-\ref{assumption_3} hold and FLTrust uses $R_l=1$ and $\beta=1$. For an arbitrary number of malicious clients, the difference between the global model learnt by FLTrust and the optimal global model $\bm{w}^*$ under no attacks is bounded. Formally, we have the following with probability at least $1 - \delta$:
	\begin{align}
	\left\| \bm{w}^t - \bm{w}^* \right\| 
	\le
	\left( 1- \rho \right)^t \left\| \bm{w}^0 - \bm{w}^* \right\| +  12\alpha\Delta_1/ \rho, \nonumber
	\end{align}
	where $\bm{w}^t$ is the global model in the $t$th iteration, $\rho = 1-  \left( \sqrt {1 - {\mu^2}/(4L^2)} + 24\alpha\Delta_2 + 2\alpha L  \right)$, $\alpha$ is the learning rate,
	$\Delta_1 = \sigma_1 \sqrt \frac{2}{\left| D_0 \right|}  \sqrt{d\log 6 + \log(3/\delta)}$,
	$\Delta_2 = {\sigma_2}\sqrt {\frac{2}{{\left| D_0 \right|}}} \sqrt {d\log \frac{18L_2}{\sigma_2} + \frac{1}{2}d\log \frac{\left| {D_0} \right|}d + \log \left( \frac{{6\sigma_2^2r\sqrt {\left| D_0 \right|} }}{\gamma_2{\sigma_1}\delta } \right)},
	$
	$\left| D_0 \right|$ is the size of the root dataset,
	$d$ is the dimension of $\bm{w}$, ${L_2} = \max \left\{ {L,{L_1}} \right\}$, and $r$ is some positive number such that $\left\| \bm{w} - {\bm{w}^*} \right\| \le r\sqrt d$ for any $\bm{w} \in \Theta$ (i.e., the parameter space $\Theta$ is constrained). When $|1- \rho| < 1$, we have $\lim_{t \rightarrow \infty} \left\| \bm{w}^t - \bm{w}^* \right\| \leq 12\alpha\Delta_1/ \rho$. 
\end{thm}	
\begin{proof}
See Appendix~\ref{sec:appendix}.
\end{proof}
}


\section{Adaptive Attacks}\label{sec:adaptive}

When an attacker knows our FLTrust is used to learn the global model, the attacker can adapt its attacks to FLTrust. Therefore, in this section, we design strong adaptive attacks to FLTrust. In particular, Fang et al. \cite{fang2019local} proposed the state-of-the-art framework that can optimize local model poisoning attacks for any given aggregation rule. We generate adaptive attacks to FLTrust via instantiating this framework with our aggregation rule. Next, we first describe the general attack framework in \cite{fang2019local}, then we discuss how to design adaptive attacks to FLTrust based on the framework.

\subsection{Local Model Poisoning Attack Framework}

The framework of local model poisoning attacks introduced in \cite{fang2019local} is general to all aggregation rules. Specifically, in each iteration of FL, the attacker aims to change the global model update the most along the opposite direction of the global model update under no attacks, by carefully crafting the local model updates on the malicious clients. Assuming the first $m$ clients are malicious. The local model poisoning attack is formulated as the following optimization problem\footnote{ Fang et al.  formulates the framework based on local models, which is equivalent to  formulating the framework based on  local model updates.}:
\begin{align}
&\max\limits_{\bm{g}_1',\bm{g}_2', \cdots,\bm{g}_m'}\bm{s}^T(\bm{g}-\bm{g}'),\nonumber\\
\text{subject to } &\bm{g} = \mathcal{A}(\bm{g}_1,\cdots,\bm{g}_m,\bm{g}_{m+1},\bm{g}_n),\nonumber\\
&\bm{g}'= \mathcal{A}(\bm{g}_1',\cdots,\bm{g}_m',\bm{g}_{m+1},\bm{g}_n),
\label{attack_framework}
\end{align}
where $\mathcal{A}$ is the aggregation rule of the FL method, $\bm{g}_i'$ is the poisoned local model update on the $i$th malicious client for $i=1,2,\cdots,m$, $\bm{g}$ is the global model update before attack, $\bm{g}'$ is the global model update after attack, and $\bm{s}$ is a column vector of the sign of the global model update before attack. 

\subsection{Adaptive Attack to Our FLTrust}

We leverage the state-of-the-art  framework to design adaptive attacks to our FLTrust. The idea is to instantiate the aggregation rule $\mathcal{A}$ with our aggregation rule in FLTrust in the framework. We denote by $\bm{e}_i= \frac{\bm{g}_i}{||\bm{g}_i||}$  the unit vector whose direction is the same as $\bm{g}_i$.  Then, our aggregation rule in Equation~(\ref{agg_local_model}) can be rewritten as follows:
\begin{align}
\label{rewrittenlocalrule}
\bm{g} = ||\bm{g}_0||\sum\limits_{i=1}^n \frac{ ReLU(c_i)}{\sum\limits_{j=1}^{n}ReLU{(c_j)}} \bm{e}_i.
\end{align}
Suppose there are $m$ malicious clients, and without loss of generality, we assume the first $m$ clients are malicious. These malicious clients send poisoned local model updates $\bm{g}_i', i=1,2,\cdots,m$ to the server. Let $\bm{e}_i'$ ($i=1,2,\cdots,m$) be the corresponding unit vectors. We note that the cosine similarity $c_i'$ between a poisoned local model update $\bm{g}_i'$ and the server model update $\bm{g}_0$ is the same as the cosine similarity between the corresponding unit vectors, i.e., $c_i'=\langle\bm{e}_i', \bm{e}_0\rangle$, where $\langle\cdot,\cdot\rangle$ means the inner product of two vectors. Therefore, we have the poisoned global model update $\bm{g}'$ under attacks as follows:
\begin{align}
\bm{g}' &= ||\bm{g}_0||\left[\sum\limits_{i=1}^m \frac{ReLU{(\langle\bm{e}_i', \bm{e}_0\rangle)}}{\sum\limits_{j=1}^{m}ReLU{(\langle\bm{e}_j', \bm{e}_0\rangle)} + \sum\limits_{j=m+1}^{n}ReLU{(c_j)}} \bm{e}_i'\right. \nonumber\\
&+  \!\! \left.\sum\limits_{i=m+1}^n \frac{ ReLU(c_i)}{\sum\limits_{j=1}^{m}ReLU{(\langle\bm{e}_j', \bm{e}_0\rangle)} + \sum\limits_{j=m+1}^{n}ReLU{(c_j)}} \bm{e}_i\right] \!\!\!.
\label{under_attack}
\end{align}

\begin{algorithm}[t]
	\caption{Our Adaptive Attack to FLTrust.}\label{attack_algo}
	\begin{algorithmic}[1]
		\renewcommand{\algorithmicrequire}{\textbf{Input:}}
		\renewcommand{\algorithmicensure}{\textbf{Output:}}
		\Require  $\bm{g}_0; \bm{g}_i$ for $i=1,2,\cdots,n; m; \sigma; \eta; \gamma; Q; V$.
		\Ensure  $\bm{e}_i'$ for $i=1,2,\cdots,m$.
		\State Compute $\bm{e}_0, \bm{e}_i, c_i$ for $i=1,2,\cdots,n$.
		\State Initialize $\bm{e}_i'$ using Trim attack for $i=1,2,\cdots,m$.
		\For {$v=1,2,\cdots,V$}
		\For {$i=1,2,\cdots,m$} 
		\For {$t=1,2,\cdots,Q$} 
		\State Randomly sample $\bm{u}\sim N(\bm{0},\sigma^2\bm{I})$.
		\State Compute $\nabla_{\bm{e}_i'} h$ according to (\ref{zero_grad}).
		\State Update $\bm{e}_i' = \bm{e}_i' + \eta \nabla_{\bm{e}_i'}h$. 
		\State Normalize $\bm{e}_i'$ such that $\left\| \bm{e}_i' \right\| =1$.
		\EndFor
		\EndFor
		\EndFor\\
		\Return $\bm{e}_i'$ for $i=1,2,\cdots,m$. 
	\end{algorithmic} 
\end{algorithm}

Substituting Equations (\ref{rewrittenlocalrule}) and (\ref{under_attack}) into (\ref{attack_framework}), and noticing that optimizing $\bm{g}_i'$ is equivalent to optimizing $\bm{e}_i'$ for $i=1,2,\cdots,m$, we can instantiate the attack framework in Equation (\ref{attack_framework}) as the following optimization problem:
\begin{align}
&\max\limits_{\bm{e}_1',\bm{e}_2',\cdots,\bm{e}_m'} \quad h(\bm{e}_1',\bm{e}_2',\cdots,\bm{e}_m'),
\label{adaptive}
\end{align}
where  $h(\bm{e}_1',\bm{e}_2',\cdots,\bm{e}_m')$ is defined as follows:
\begin{align}
& h(\bm{e}_1',\bm{e}_2',\cdots,\bm{e}_m') 
                = ||\bm{g}_0||\bm{s}^T \left[\sum\limits_{i=1}^n \frac{ ReLU(c_i)}{\sum\limits_{j=1}^{n}ReLU{(c_j)}} \bm{e}_i \right. 
\nonumber\\
&\;\left.- \sum\limits_{i=1}^b \frac{ReLU{(\langle\bm{e}_i', \bm{e}_0\rangle)}}{\sum\limits_{j=1}^{m}ReLU{(\langle\bm{e}_j', \bm{e}_0\rangle)} + \sum\limits_{j=m+1}^{n}ReLU{(c_j)}} \bm{e}_i' \right.  
\nonumber\\
&\; \left. - \!\!\!\! \left.\sum\limits_{i=m+1}^n \frac{ ReLU(c_i)}{\sum\limits_{j=1}^{m}ReLU{(\langle\bm{e}_j', \bm{e}_0\rangle)} + \sum\limits_{j=m+1}^{n}ReLU{(c_j)}} \bm{e}_i \right.   \right] \!\!, 
\end{align}
where $\bm{s}^T=\text{sgn}(\bm{g})^T$ is the sign of the global model update without attacks. Solving the optimization problem generates an adaptive attack to FLTrust. We consider a strong adaptive attacker who has full knowledge about the FL system when solving the optimization problem. In particular, $||\bm{g}_0||, \bm{s}, c_i\ (i=1,2,\cdots,n), \bm{e}_0$, and $\bm{e}_i\ (i=1,2,\cdots,n)$ are all available to the attacker. 

\myparatight{Solving the optimization problem} We use a standard gradient ascent approach to solve the optimization problem. Specifically, we can compute the gradient $\nabla_{\bm{e}_i'}h$ of the objective function $h$ with respect to each $\bm{e}_i'$ and move $\bm{e}_i'$ a small step along the gradient. Since the gradient $\nabla_{\bm{e}_i'}h$  involves a Jacobian matrix of $\bm{e}_i'$, it is not practical to directly compute the gradient. Therefore, we leverage a zeroth-order method \cite{cheng2018query,nesterov2017random} to compute the gradient, which is a standard method to solve such optimization problems with computationally intractable objective functions. Specifically, we compute the gradient 
$\nabla_{\bm{e}_i'} h$  as follows:
\begin{align}
\label{zero_grad}
\nabla_{\bm{e}_i'} h \approx \frac{h(\bm{e}_i' + \gamma\bm{u}) - h(\bm{e}_i') }{\gamma} \cdot \bm{u},
\end{align}
where $\bm{u}$ is a random vector sampled from the multivariate Gaussian distribution $N(\bm{0},\sigma^2\bm{I})$ with zero mean and diagonal covariance matrix, and $\gamma > 0$ is a smoothing parameter. 

We optimize $\bm{e}_i'$ one by one following the standard {coordinate ascent} approach, i.e., when optimizing $\bm{e}_i'$, all other $\bm{e}_j', j\neq i$ are fixed. Specifically, we use projected gradient ascent to iteratively optimize $\bm{e}_i'$. In the beginning, we initialize $\bm{e}_i'$ using  the Trim attack, i.e., we use the Trim attack to compute the poisoned local model updates and initialize $\bm{e}_i'$ as the corresponding unit vector. Then, in each iteration, we sample a random vector $\bm{u}$ from $N(\bm{0},\sigma^2\bm{I})$ and compute the gradient $\nabla_{\bm{e}_i'}h$ following Equation (\ref{zero_grad}). We multiply the  gradient by a step size $\eta$ and add it to $\bm{e}_i'$ to get the new $\bm{e}_i'$. Finally, we project $\bm{e}_i'$ to the unit sphere to ensure that $\bm{e}_i'$ is a valid unit vector. We repeat the gradient ascent process  for $Q$ iterations. Moreover, we repeat the iterations over the unit vectors for $V$ iterations. Algorithm \ref{attack_algo} shows our adaptive attack. We let $\bm{g}_i'=\Vert\bm{g}_0\Vert\cdot\bm{e}_i'$ after $\bm{e}_i'$ is solved for $i=1,2,\cdots,m$.


\section{Evaluation} 
We evaluate our FLTrust against both existing poisoning attacks to FL and adaptive attacks in this section. 

\subsection{Experimental Setup}
\label{sec:setup}

\subsubsection{Datasets} 
We use multiple datasets from different domains in our evaluation, including five image classification datasets and a human activity recognition dataset. \xc{We follow  previous work \cite{fang2019local} to distribute the training examples in a dataset among clients. 
Assuming there are $M$ classes in a dataset. We randomly split the clients into $M$ groups. A training example with label $l$ is assigned to group $l$ with probability $q>0$ and to any other group with probability $\frac{1-q}{M-1}$. Within the same group, data are uniformly distributed to each client.  $q$ controls the distribution difference of the clients' local training data. If $q=1/M$, then the clients' local training data are independent and identically distributed (IID), otherwise the clients' local training data are non-IID. Moreover, a larger $q$ indicates a higher degree of non-IID among the clients' local training data. One characteristic of FL is that clients often have non-IID local training data~\cite{Konen16,McMahan17}. Therefore, we will set $q>1/M$ by default to simulate the non-IID settings. 

Next, we use the MNIST dataset as an example to show the distribution process. Assume we have 100 clients in total and set $q=0.5$.  $M=10$ for the MNIST dataset. We first randomly split the clients into 10 groups, each containing 10 clients. For a training image of digit $l$ (e.g., $l=5$), we first assign it to group 5 with probability 0.5, and to any other group with probability $\frac{1-0.5}{10-1}\approx 0.056$. Once the group is determined, e.g., group 5 is chosen, we will select a client from group 5  uniformly at random and assign this training image to the selected client. }

\myparatight{MNIST-0.1} MNIST \cite{lecun2010mnist} is a 10-class digit image classification dataset, which consists of  60,000 training examples and 10,000 testing examples. We set $q=0.1$ in MNIST-0.1, which indicates local training data are IID among clients. We use MNIST-0.1 to show that FLTrust is also effective in the IID setting. 

\myparatight{MNIST-0.5} In MNIST-0.5, we simulate non-IID local training data among the clients via setting $q=0.5$. 

\myparatight{Fashion-MNIST} Fashion-MNIST \cite{xiao2017/online} is a 10-class fashion image classification task, which has a predefined training set of 60,000 fashion images and a testing set of 10,000 fashion images. Like the MNIST-0.5 dataset, we distribute the training examples to the clients with $q=0.5$ to simulate non-IID local training data. 

\myparatight{CIFAR-10} CIFAR-10 \cite{krizhevsky2009learning} is a color image classification dataset consisting of predefined 50,000 training examples and 10,000 testing examples. Each example belongs to one of the 10 classes. To simulate non-IID local training data, we distribute the training examples to clients with $q=0.5$. 

\myparatight{Human activity recognition (HAR)} The HAR dataset \cite{anguita2013public} consists of human activity data collected from the smartphones of 30 real-world users. The data are  signals from multiple sensors on a user's smartphone, and the task is to  predict the user's activity among 6 possible activities, i.e., WALKING, WALKING\_UPSTAIRS, WALKING\_DOWNSTAIRS, SITTING, STANDING, and LAYING. Each example includes 561 features and there are 10,299 examples in total. Unlike the previous datasets, we don't need to distribute the data to clients in this dataset, as each user is naturally considered as a client. \xc{HAR represents a real-world FL scenario, where each user is considered as a client. We use 75\% of each client's data as training examples and the rest 25\% as testing examples. We note that HAR has unbalanced local training data on clients: the maximum number of training examples on a client is 409, the minimum number is 281, and the mean is 343.} 

\begin{table*}[!t]\renewcommand{\arraystretch}{1}
\caption{The default FL system parameter settings.}
\centering
\vspace{2mm}
\addtolength{\tabcolsep}{-3pt}
\begin{tabular}{|c|c|c|c|c|c|c|c|} \hline 
{} & {Explanation} & {MNIST-0.1} & {MNIST-0.5} & {Fashion-MNIST} & {CIFAR-10} & {HAR} & {\xc{CH-MNIST}} \\ \hline
{$n$} & {\# clients} & \multicolumn{4}{c|}{100} & {30} & {\xc{40}}\\ \hline
{$\tau$} & {\# clients selected in each iteration} & \multicolumn{6}{c|}{$n$}\\ \hline
{$R_l$} & {\# local iterations} & \multicolumn{6}{c|}{1}\\ \hline
{$R_g$} & {\# global iterations} & \multicolumn{2}{c|}{2,000} & {2,500} & {1,500} & {1,000} & {\xc{2,000}}\\ \hline
{$b$} & {batch size} & \multicolumn{3}{c|}{32} & {64} & \multicolumn{2}{c|}{32}\\ \hline
{$\alpha\cdot\beta$} & {combined learning rate} & \multicolumn{2}{c|}{$3\times10^{-4}$} & {$6\times10^{-3}$} & {$2\times10^{-4}$} & {$3\times10^{-3}$}  & \xc{\makecell{$3\times10^{-4}$ (decay at the 1500th and \\ 1750th iterations with factor 0.9)}}\\ \hline
{$m/n$} & {fraction of malicious clients (\%)} & \multicolumn{6}{c|}{20}\\ \hline
{$m$} & {\# malicious clients} & \multicolumn{4}{c|}{20} & {6} & {\xc{8}}\\ \hline
{$f$} & {Krum parameter} & \multicolumn{6}{c|}{$m$}\\ \hline
{$k$} & {Trim-mean parameter} & \multicolumn{6}{c|}{$m$}\\ \hline
{$|D_0|$} & {size of the root dataset} & \multicolumn{6}{c|}{100}\\ \hline
\end{tabular}
\label{tab:param}
\vspace{-3mm}
\end{table*}

\xc{\myparatight{CH-MNIST} CH-MNIST \cite{kather2016multi} is a medical image classification dataset consisting of 5,000 images of histology tiles collected from colorectal cancer patients. Each example has $64\times 64$ gray-scale pixels and belongs to one of the 8 classes. We use 4,000 images selected randomly as the training examples and use the other 1,000 images as the testing examples. To simulate non-IID local training data, we distribute the training examples to clients with $q=0.5$.}

\subsubsection{Evaluated Poisoning Attacks} 

We consider both data poisoning attacks and local model poisoning attacks. For data poisoning attack, we consider the popular label flipping attack. For local model poisoning attacks, we evaluate Krum attack, Trim attack, and our adaptive attack (untargeted attacks) \cite{fang2019local},  as well as Scaling attack (targeted attack) \cite{bagdasaryan2020backdoor}.

\myparatight{Label flipping (LF) attack} We use the same label flipping attack setting as \cite{fang2019local}. In particular, for each training example on the malicious clients, we flip its label $l$ to $M-l-1$, where $M$ is the total number of labels and $l\in\{0,1,\cdots,M-1\}$. 

\myparatight{Krum attack} Krum attack is an untargeted local model poisoning attack optimized for the Krum aggregation rule. We use the default parameter settings in \cite{fang2019local} for the Krum attack. 

\myparatight{Trim attack} Trim attack is an untargeted local model poisoning attack optimized for the Trim-mean and Median aggregation rules. We use the default parameter settings in \cite{fang2019local} for the Trim attack. 

\myparatight{Scaling attack} Scaling attack is a targeted local model poisoning attack. Specifically, we consider the attacker-chosen target testing examples are normal testing examples with a predefined feature-pattern trigger embedded. Following \cite{bagdasaryan2020backdoor}, we use the data augmentation scheme in \cite{Gu17} to implement the Scaling attack. Specifically,  each malicious client copies $p$ fraction of its local training examples, embeds the  trigger to them, changes their labels to the attacker-chosen target label, and uses them to augment its local training data.  Then, in each iteration of FL, each malicious client computes its local model update based on the augmented local training data and scales it by a factor $\lambda\gg1$ before sending it to the server. 

Specifically, we use the same pattern trigger in \cite{Gu17} as our trigger for MNIST-0.1, MNIST-0.5, Fashion-MNIST, \xc{and CH-MNIST}, and we set the attacker-chosen target label as 0; for CIFAR-10, we consider the same pattern trigger and target label (i.e., ``bird") in \cite{bagdasaryan2020backdoor}; 
and for HAR, we create a feature-pattern trigger by setting every 20th feature to 0 and we set the target label as ``WALKING\_UPSTAIRS". Following previous work~\cite{bagdasaryan2020backdoor}, we set the scaling factor $\lambda=n$, where $n$ is the number of clients. In each dataset, the attacker-chosen target testing examples consist of the trigger-embedded normal testing examples whose true labels are not the target label. 

\xc{\myparatight{Adaptive attack}
We evaluate the adaptive attack proposed in Section \ref{sec:adaptive}.  Our adaptive attack leverages an zeroth-order optimization method. Following the suggestions by previous work \cite{cheng2018query,nesterov2017random}, we set $\sigma^2=0.5$ and $\gamma = 0.005$ in the zeroth-order method. Moreover, we set $\eta=0.01$ and $V=Q=10$ so that the adaptive attack converges. }

\subsubsection{Evaluation Metrics} For the LF attack, Krum attack, Trim attack, \xc{and adaptive attack}, we use the standard \emph{testing error rate} of the global model to evaluate an FL method since these attacks aim to increase the testing error rate. Specifically, the testing error rate of a global model is the fraction of testing examples whose labels are incorrectly predicted by the global model.  An FL method is more robust against these attacks if its global models achieve lower testing error rates under these attacks. The Scaling attack is a targeted attack, which aims to preserve the testing error rate of normal testing examples while making the global model predict the attacker-chosen target label for the attacker-chosen target testing examples. Therefore, other than the testing error rate, we further use \emph{attack success rate} to measure the Scaling attack. Specifically, the attack success rate is the fraction of the attacker-chosen target testing examples whose labels are predicted as the attacker-chosen target label by the global model. An FL method is more robust against the Scaling attack if its global model achieves a lower attack success rate. 

\begin{table}[!t]
\caption{The CNN architecture of the global model used for MNIST-0.1, MNIST-0.5, and Fashion-MNIST.}
\centering
\vspace{2mm}
\begin{tabular}{|c|c|} \hline 
{Layer} & {Size} \\ \hline
{Input} & { $28\times28\times1$}\\ \hline
{Convolution + ReLU} & { $3\times3\times30$}\\ \hline
{Max Pooling} & { $2\times2$}\\ \hline
{Convolution + ReLU} & { $3\times3\times50$}\\ \hline
{Max Pooling} & { $2\times2$}\\ \hline
{Fully Connected + ReLU} & {100}\\ \hline
{Softmax} & {10}\\ \hline
\end{tabular}
\label{tab:cnn}
\vspace{-2mm}
\end{table}

\subsubsection{FL System Settings} By default, we assume there are $n=100$ clients in total for each dataset except HAR and CH-MNIST. For HAR, the data are collected from 30 users, each of which is treated as a client. Therefore, HAR has 30 clients in total. \xc{For CH-MNIST, there are only 4,000 training examples in total and thus we assume 40 clients such that each client  has 100 training examples on average. Unless otherwise mentioned, we assume 20\% of the clients are malicious for each dataset. However, we will also explore the impact of the fraction of malicious clients.} Table \ref{tab:param} shows the default FL system settings that we will use unless otherwise mentioned.

\myparatight{Global models} We train different types of global models on different datasets to show the generality of our method. Specifically, for MNIST-0.1, MNIST-0.5, and Fashion-MNIST, we train a convolutional neural network (CNN) as the global model.  Table \ref{tab:cnn} shows the architecture of the CNN.  And we train a logistic regression (LR) classifier as the global model for HAR. \xc{For CIFAR-10 and CH-MNIST, we consider the widely used ResNet20 architecture \cite{he2016deep} as the global model.}

\myparatight{Parameter settings of the FL methods}  We compare FLTrust with FedAvg~\cite{Konen16,McMahan17}, Krum~\cite{Blanchard17}, Trim-mean~\cite{Yin18}, and Median~\cite{Yin18}. 
Details of these FL methods can be found in Section \ref{sec:background}. FedAvg is a popular FL method in non-adversarial settings, while Krum, Trim-mean, and Median are Byzantine-robust FL methods.  These methods all follow the three-step framework described in Algorithm \ref{global_upate_algo}, though they use different aggregation rules. Therefore, they all use the parameters $\tau$, $R_l$, $R_g$, $\alpha$, $\beta$, and $b$.   Following previous work \cite{fang2019local}, we set $\tau=n$, i.e., all clients are selected in each iteration; and we set $R_l=1$, in which we can treat the product of the global learning rate $\alpha$ and the local learning rate $\beta$ as a single learning rate. We set this combined learning rate on each dataset to achieve small training error rates and fast convergence. We set the batch size $b=32$ for all datasets except CIFAR-10, where we set $b=64$. We set the number of global iterations $R_g$ such that the FL methods converge. Specifically, $R_g=2,000$ for MNIST-0.1, MNIST-0.5, \xc{and CH-MNIST}; $R_g=2,500$ for Fashion-MNIST; $R_g=1,500$ for CIFAR-10; and $R_g=1,000$ for HAR. 

Krum further has the parameter $f$ and Trim-mean further has the trim parameter $k$, both of which are an upper bound of the number of malicious clients. We set $f=k=m$, which assumes that the server knows the exact number of malicious clients and gives advantages to Krum and Trim-mean. 

\myparatight{Root dataset} Our FLTrust requires a small root dataset. By default, we assume the root dataset has only 100 training examples. Moreover, we consider the following two cases depending on how the root dataset is created. 
\begin{itemize}
    \item {\bf Case I.} We assume the service provider can collect a representative root dataset for the learning task, i.e., the root dataset has the same distribution as the overall training data distribution of the learning task. In particular,  we sample the root dataset  from the union of the clients' clean local training data uniformly at random. For instance, for MNIST-0.5, we sample the root dataset  from its 60,000 training examples uniformly at random. 
    \item {\bf Case II.} We assume the root dataset has a distribution different from the overall training data distribution of the learning task. \xc{In particular, we assume the root dataset is biased towards a certain class. Specifically, we sample a fraction of the examples in the root dataset from a particular class (class 1 in our experiments) in the union of the clients' clean local training data and the remaining examples are sampled from the remaining classes uniformly at random, where we call the fraction \emph{bias probability}.  Note that, for all the datasets except HAR and CH-MNIST, the root dataset has the same distribution as the overall training data, i.e., Case II reduces to Case I, when the bias probability is 0.1 because they have 10 classes; for HAR and CH-MNIST,  Case II reduces to Case I when the bias probability is 0.17 and 0.125 because they have 6 and 8 classes, respectively. The root data distribution deviates more from the overall training data distribution when the bias probability is larger. }
\end{itemize}

In both cases, we exclude the sampled root dataset from the clients' local training data, indicating that the root dataset is collected independently by the service provider. Unless otherwise mentioned, we consider Case I. 

\begin{table}[!t]\renewcommand{\arraystretch}{1}
\centering
\caption{The testing error rates of different FL methods under different attacks and  the attack success rates of the Scaling attacks. The results for the Scaling attacks are in the form of ``testing error rate / attack success rate''.}
\vspace{1mm}
\centering
\addtolength{\tabcolsep}{-4pt}
\captionsetup[subfloat]{captionskip=0pt}

\subfloat[CNN global model, MNIST-0.1]{
\begin{tabular}{|c|c|c|c|c|c|}
	\hline
	& FedAvg	& Krum & Trim-mean  & Median  & FLTrust  \\
	\hline
	No attack & 0.04 & 0.10 & 0.06 & 0.06   & 0.04  \\
	\hline
	LF attack & 0.06 &  0.10    &   0.05  & 0.05 & 0.04  \\
	\hline
	Krum attack & 0.10 &  0.90   &  0.07    & 0.07 & 0.04  \\
	\hline
	Trim attack & 0.16  & 0.10    &  0.13    & 0.13 &  0.04 \\
	\hline
	Scaling attack  & 0.02 / 1.00  & 0.10 / 0.00 & 0.05 / 0.01  & 0.05 / 0.01 & 0.03 / 0.00 \\
    \hline
    \xc{Adaptive attack} & \xc{0.08}  & \xc{0.10} & \xc{0.11}  & \xc{0.13} & \xc{0.04} \\
    \hline
\end{tabular}
}

\subfloat[CNN global model, MNIST-0.5]{
\begin{tabular}{|c|c|c|c|c|c|}
	\hline
	& FedAvg & Krum & Trim-mean  & Median  & FLTrust  \\
	\hline
	No attack & 0.04 & 0.10 & 0.06 & 0.06   & 0.05  \\
	\hline
	LF attack & 0.06 &  0.10    &   0.06  & 0.06 & 0.05  \\
	\hline
	Krum attack & 0.10 &  0.91   &  0.14    & 0.15 & 0.05  \\
	\hline
	Trim attack  & 0.28 & 0.10    &  0.23    & 0.43 &  0.06 \\
	\hline
	Scaling attack  & 0.02 / 1.00  & 0.09 / 0.01 & 0.06 / 0.02  & 0.06 / 0.01 & 0.05 / 0.00 \\
    \hline
    \xc{Adaptive attack} & \xc{0.13}  & \xc{0.10} & \xc{0.22}  & \xc{0.90} & \xc{0.06} \\
    \hline
\end{tabular}
}

\subfloat[CNN global model, Fashion-MNIST]{
\begin{tabular}{|c|c|c|c|c|c|}
	\hline
	& FedAvg & Krum & Trim-mean  & Median  & FLTrust  \\
	\hline
	No attack & 0.10  & 0.16 & 0.14 & 0.14   & 0.11  \\
	\hline
	LF attack & 0.14 &  0.15    &   0.26  & 0.21 & 0.11  \\
	\hline
	Krum attack & 0.13 &  0.90   &  0.18    & 0.23 & 0.12  \\
	\hline
	Trim attack  & 0.90 & 0.16    &  0.24    & 0.27 &  0.14 \\
	\hline
	Scaling attack  & 0.90 / 1.00  &  0.16 / 0.03   &  0.17 / 0.85    & 0.16 / 0.05 &  0.11 / 0.02 \\
    \hline
    \xc{Adaptive attack} & \xc{0.90}  & \xc{0.18} & \xc{0.34}  & \xc{0.24} & \xc{0.14} \\
    \hline
\end{tabular}
}

\subfloat[ResNet20 global model, CIFAR-10]{
	\begin{tabular}{|c|c|c|c|c|c|}
		\hline
		& FedAvg & Krum & Trim-mean  & Median  & FLTrust  \\
		\hline
		No attack & 0.16 & 0.54 & 0.24 & 0.25   & 0.18  \\
		\hline
		LF attack & 0.21 &  0.56  &  0.27   & 0.45 & 0.18  \\
		\hline
		Krum attack & 0.24 &  0.90   &  0.52   &  0.64 &  0.18 \\
		\hline
		Trim attack & 0.81 &  0.51   &   0.72   & 0.75 & 0.20 \\
		\hline
		Scaling attack  & 0.90 / 1.00  &  0.44 / 0.07   &  0.22 / 0.96   & 0.25 / 0.96 & 0.18 / 0.02 \\
	    \hline
	    \xc{Adaptive attack} & \xc{0.90}  & \xc{0.58} & \xc{0.69}  & \xc{0.82} & \xc{0.20} \\
        \hline
	\end{tabular}
}

\subfloat[LR global model, HAR]{
\begin{tabular}{|c|c|c|c|c|c|}
	\hline
	& FedAvg & Krum & Trim-mean  & Median  & FLTrust  \\
	\hline
	No attack & 0.03 & 0.12 & 0.04 & 0.05   & 0.04  \\
	\hline
	LF attack & 0.17 &  0.10    &   0.05  & 0.05 & 0.04  \\
	\hline
	Krum attack & 0.03  & 0.22   &  0.05    & 0.05 & 0.04  \\
	\hline
	Trim attack  & 0.32 & 0.10    &  0.36    & 0.13 &  0.05 \\
	\hline
	Scaling attack  & 0.04 / 0.81  & 0.10 / 0.03    &  0.04 / 0.36    & 0.05 / 0.13 &  0.05 / 0.01 \\
	\hline
	\xc{Adaptive attack} & \xc{0.04}  & \xc{0.19} & \xc{0.05}  & \xc{0.06} & \xc{0.05} \\
    \hline
\end{tabular}
}

\xc{
\subfloat[\xc{ResNet20 global model, CH-MNIST}]{
\begin{tabular}{|c|c|c|c|c|c|}
	\hline
	& FedAvg & Krum & Trim-mean  & Median  & FLTrust  \\
	\hline
	No attack & 0.10 & 0.24 & 0.10 & 0.11   & 0.10  \\
	\hline
	LF attack & 0.12 &  0.39    &   0.15  & 0.13 & 0.12  \\
	\hline
	Krum attack & 0.11  & 0.95   &  0.13    & 0.13 & 0.12  \\
	\hline
	Trim attack  & 0.64 & 0.21    &  0.55    & 0.44 &  0.13 \\
	\hline
	Scaling attack  & 0.26 / 0.20  & 0.34 / 0.03    &  0.14 / 0.02    & 0.11 / 0.01 &  0.14 / 0.03 \\
	\hline
	Adaptive attack & 0.14  & 0.29 & 0.50  & 0.47 & 0.13 \\
    \hline
\end{tabular}
}
}

\label{tab:effective}
\vspace{-4mm}
\end{table}

\subsection{Experimental Results}\label{sec:exp_res}

\myparatight{Our FLTrust achieves the three defense goals} Recall that we have three defense goals (discussed in Section~\ref{sec:prob}): fidelity, robustness, and efficiency. Table \ref{tab:effective} shows the testing error rates of different FL methods under different attacks \xc{including our adaptive attack}, as well as the attack success rate of the Scaling attack on the six datasets. Our results show that FLTrust achieves the three goals. 

First, when there is no attack, our FLTrust has testing error rates similar to FedAvg, achieving the fidelity goal. However, existing Byzantine-robust FL methods may have higher or much higher testing error rates under no attacks. For instance, on MNIST-0.1, the testing error rates for FedAvg and FLTrust are both 0.04, while they are 0.10, 0.06, and 0.06 for Krum, Trim-mean, and  Median, respectively; on CH-MNIST, FedAvg, Trim-mean, and FLTrust achieve testing error rates 0.10, while Krum and Median achieve testing error rates 0.24 and 0.11, respectively. Our results indicate that  FLTrust is more accurate than  existing Byzantine-robust FL methods in non-adversarial settings. This is because existing Byzantine-robust FL methods exclude some local model updates when aggregating them as the global model update, while FLTrust considers all of them with the help of the root dataset. 

Second, our FLTrust achieves the robustness goal, while existing FL methods do not. Specifically, \xc{the testing error rates of FLTrust under the untargeted attacks including our adaptive attack are at most 0.04 higher than those of FedAvg under no attacks on the six datasets.}  On the contrary, every existing Byzantine-robust FL method has much higher testing error rates, especially under the untargeted attack that is optimized for the method. For instance, on MNIST-0.5, Krum attack increases the testing error rate of Krum from 0.10 to 0.91, while Trim attack increases the testing error rates of Trim-mean and Median from 0.06 to 0.23 and 0.43, respectively.  FedAvg may have lower testing error rates than  existing Byzantine-robust FL methods under the evaluated untargeted attacks. This is because these untargeted attacks are not optimized for FedAvg. Previous work~\cite{Blanchard17} showed that FedAvg can be arbitrarily manipulated by a single malicious client. 

\begin{table*}[!t]\renewcommand{\arraystretch}{1.1}
	\centering
	\caption{The testing error rates of different variants of FLTrust under different attacks and the attack success rates of the Scaling attacks on MNIST-0.5. The results for the Scaling attacks are in the form of ``testing error rate / attack success rate''. ``--'' means that the attacks are not applicable.}
	\vspace{2mm}
	\addtolength{\tabcolsep}{0pt}
	\begin{tabular}{|c|c|c|c|c|c|c|}
		\hline
		& No attack  & LF attack  & Krum attack  & Trim attack & Scaling attack & \xc{Adaptive attack} \\
		\hline
		FLTrust-Server & 0.21     & --   &  --  & -- & -- & -- \\
		\hline
		FLTrust-withServer & 0.07    & 0.08    &  0.09  & 0.10 & 0.08 / 0.01 & \xc{0.94} \\
		\hline
		FLTrust-NoReLU & 0.28   & 0.90 & 0.90 & 0.90 & 0.94 / 0.08 & \xc{0.90}\\
		\hline
		FLTrust-NoNorm & 0.05     & 0.06    &  0.06  & 0.08 & 0.94 / 0.08 & \xc{0.06}\\
		\hline
		FLTrust-ParNorm & 0.06    & 0.06    &  0.06  & 0.06 & 0.06 / 0.01 & \xc{0.06}\\
		\hline
		FLTrust & 0.05     & 0.05    &  0.05  & 0.06 & 0.05 / 0.00 & \xc{0.06}\\
		\hline
	\end{tabular}%
	\label{tab:variants}
\end{table*}%

\begin{figure}[!t]
	\centering
	\includegraphics[scale = 0.2]{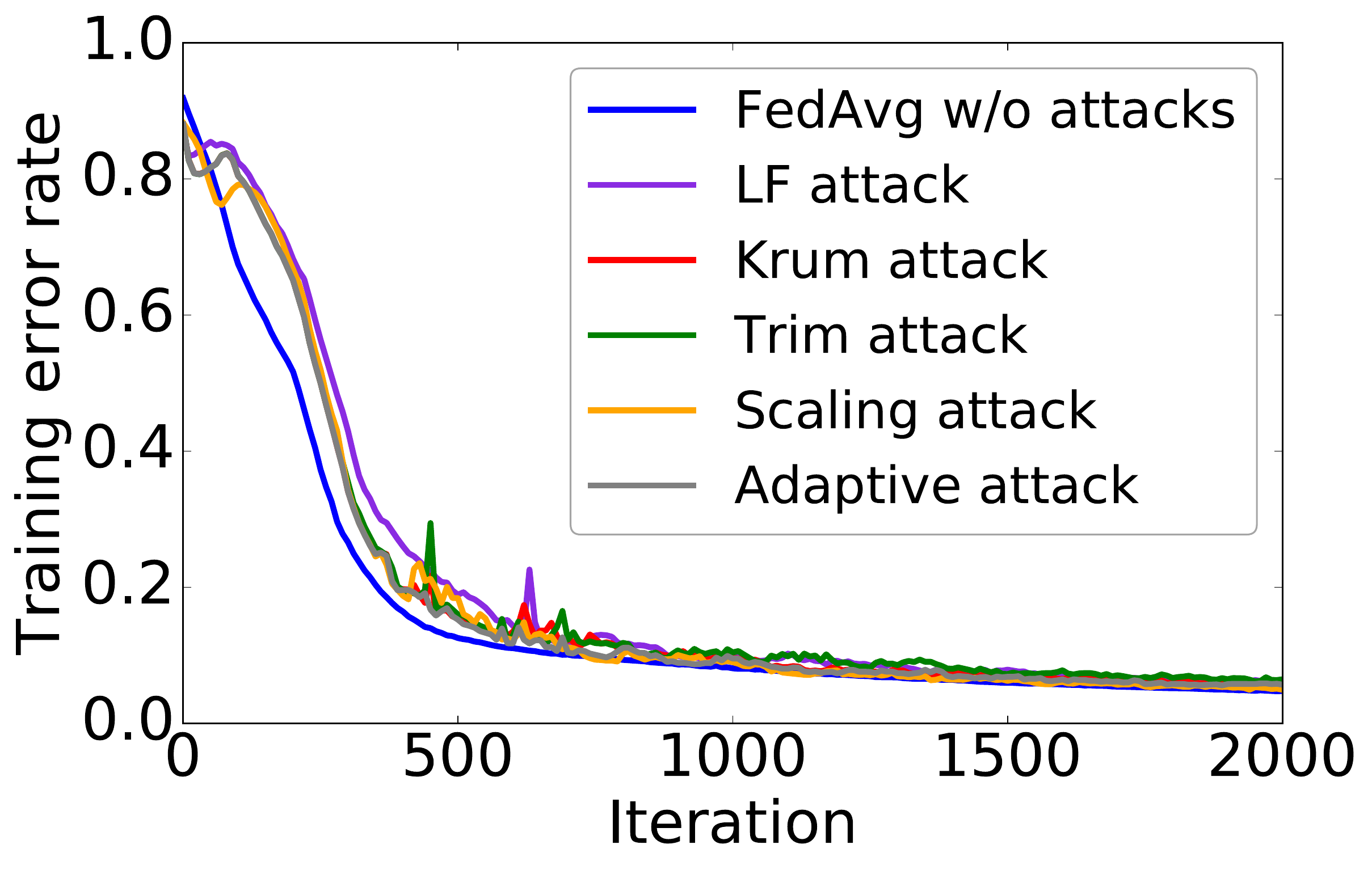}
	\caption{\xc{The training error rates vs. the number of iterations for FLTrust under different attacks and FedAvg without attacks on MNIST-0.5.}}
	\label{fig:convergence}
	\vspace{-2mm}
\end{figure}

Moreover,  for the Scaling attack, FLTrust substantially reduces its attack success rates. Specifically, the attack success rates for FLTrust are at most 0.03. On the contrary, the attack success rates for FedAvg are always high on the six datasets, and they are also high for the existing Byzantine-robust FL methods on multiple datasets, indicating that existing FL methods are not robust against the Scaling attack. One interesting observation is that the Scaling attack may decrease the testing error rates in some cases. \xc{We suspect the reason may be that the data augmentation in the Scaling attack positively impacts the aggregation of the local model updates. Specifically, the data augmentation in the Scaling attack improves the diversity of the training data, and thus helps the learned global model better generalize to the testing dataset.
}

Third, FLTrust achieves the efficiency goal. Specifically, in each iteration, FLTrust does not incur extra overhead to the clients; and compared to FedAvg, the extra computation incurred to the server by FLTrust includes computing a server model update, computing the trust scores, and normalizing the local model updates, which are negligible for the powerful server. Moreover, Figure \ref{fig:convergence} shows the training error rates versus the global iteration number for  FLTrust under different attacks and FedAvg under no attack on MNIST-0.5. Our results show that FLTrust converges as fast as FedAvg, which means that FLTrust also does not incur extra communications cost for the clients (each iteration of FL requires communications between clients and server), compared to FedAvg under no attacks. We note that Krum, Trim-mean, and Median do not incur extra overhead to the clients. However, Krum incurs significant computational overhead to the server when there are a large number of clients. This is because Krum requires calculating pairwise distance between local model updates in each iteration. 

\begin{table}[!t]\renewcommand{\arraystretch}{1.1}
	\centering
	\caption{The testing error rates of FLTrust under different attacks and the attack success rates of the Scaling attacks when the root dataset is sampled with different bias probabilities in Case II.}
	\vspace{2mm}
	\addtolength{\tabcolsep}{-5pt}
	\xc{
	\subfloat[MNIST-0.1]{
        \begin{tabular}{|c|c|c|c|c|c|c|}
		\hline
		{Bias probability} & {0.1}  & {0.2} & {0.4} & {0.6} & {0.8}  & {1.0}\\
		\hline
		{No attack} & {0.04}  & {0.04} & {0.04} & {0.05} & {0.05}  & {0.34}\\
		\hline
		{LF attack} & {0.04}  & {0.04} & {0.04} & {0.05} & {0.78}  & {0.84}\\
		\hline
		{Krum attack} & {0.04}  & {0.04} & {0.07} & {0.89} & {0.89}  & {0.89}\\
		\hline
		{Trim attack} & {0.04}  & {0.05} & {0.08} & {0.12} & {0.46}  & {0.89}\\
		\hline
		{Scaling attack} & {0.03 / 0.00}  & {0.03 / 0.01} & {0.04 / 0.00} & {0.04 / 0.00} & {0.06 / 0.01}  & {0.42 / 0.01}\\
		\hline
		{Adaptive attack} & {0.04}  & {0.05} & {0.08} & {0.12} & {0.90}  & {0.90}\\
		\hline
	    \end{tabular}%
        \label{mnist-0.1-bias}
    }}
    
	\subfloat[MNIST-0.5]{
        \begin{tabular}{|c|c|c|c|c|c|c|}
		\hline
		{Bias probability} & {0.1}  & {0.2} & {0.4} & {0.6} & {0.8}  & {1.0}\\
		\hline
		{No attack} & {0.05}  & {0.05} & {0.06} & {0.08} & {0.11}  & {0.80}\\
		\hline
		{LF attack} & {0.05}  & {0.05} & {0.08} & {0.10} & {0.25}  & {0.89}\\
		\hline
		{Krum attack} & {0.05}  & {0.05} & {0.08} & {0.12} & {0.86}  & {0.89}\\
		\hline
		{Trim attack} & {0.06}  & {0.06} & {0.08} & {0.12} & {0.16}  & {0.89}\\
		\hline
		{Scaling attack} & {0.05 / 0.00}  & {0.05 / 0.01} & {0.06 / 0.00} & {0.07 / 0.01} & {0.12 / 0.00}  & {0.86 / 0.01}\\
		\hline
		\xc{Adaptive attack} & \xc{0.06}  & \xc{0.07} & \xc{0.08} & \xc{0.13} & \xc{0.90}  & \xc{0.90}\\
		\hline
	    \end{tabular}%
        \label{mnist-0.5-bias}
    }

    \xc{
    \subfloat[Fashion-MNIST]{
        \begin{tabular}{|c|c|c|c|c|c|c|}
		\hline
		{Bias probability} & {0.1}  & {0.2} & {0.4} & {0.6} & {0.8}  & {1.0}\\
		\hline
		{No attack} & {0.11}  & {0.11} & {0.12} & {0.15} & {0.16}  & {0.90}\\
		\hline
		{LF attack} & {0.11}  & {0.11} & {0.12} & {0.12} & {0.14}  & {0.90}\\
		\hline
		{Krum attack} & {0.12}  & {0.12} & {0.16} & {0.90} & {0.90}  & {0.90}\\
		\hline
		{Trim attack} & {0.14}  & {0.14} & {0.15} & {0.21} & {0.90}  & {0.90}\\
		\hline
		{Scaling attack} & {0.11 / 0.02}  & {0.12 / 0.04} & {0.12 / 0.04} & {0.13 / 0.02} & {0.15 / 0.03}  & {0.90 / 0.00}\\
		\hline
		{Adaptive attack} & {0.14}  & {0.14} & {0.16} & {0.90} & {0.90}  & {0.90}\\
		\hline
	    \end{tabular}%
        \label{fashionmnist-bias}
    }}

     \xc{
    \subfloat[CIFAR-10]{
        \begin{tabular}{|c|c|c|c|c|c|c|}
		\hline
		{Bias probability} & 0.1  & 0.2 & 0.4 & 0.6  & 0.8 & 1.0 \\
		\hline
		{No attack} & 0.18  & 0.18 & 0.18 & 0.21 & 0.90 & 0.90 \\
		\hline
		{LF attack} & 0.18  & 0.19 & 0.20  & 0.24 & 0.90 & 0.90 \\
		\hline
		{Krum attack} & 0.18  & 0.18 & 0.19 & 0.33  & 0.90  & 0.90\\
		\hline
		{Trim attack} & 0.20 & 0.20  &  0.24 & 0.63 & 0.90  & 0.90\\
		\hline
		{Scaling attack} & 0.18 / 0.02  & 0.18 / 0.00 & 0.18 / 0.03 & 0.22 / 0.04  &  0.90 / 0.00 & 0.90 / 0.00 \\
		\hline
		{Adaptive attack} & 0.20  & 0.20 & 0.27 & 0.68 & 0.90 & 0.90\\
		\hline
	    \end{tabular}%
        \label{CIFAR-10-bias}
    }}
    
    \xc{
    \subfloat[HAR]{
        \begin{tabular}{|c|c|c|c|c|c|c|}
		\hline
		{Bias probability} & {0.17}  & {0.2} & {0.4} & {0.6} & {0.8}  & {1.0}\\
		\hline
		{No attack} & {0.04}  & {0.04} & {0.06} & {0.06} & {0.07}  & {0.48}\\
		\hline
		{LF attack} & {0.04}  & {0.05} & {0.06} & {0.05} & {0.07}  & {0.48}\\
		\hline
		{Krum attack} & {0.04}  & {0.05} & {0.05} & {0.05} & {0.09}  & {0.48}\\
		\hline
		{Trim attack} & {0.05}  & {0.05} & {0.06} & {0.09} & {0.14}  & {0.48}\\
		\hline
		{Scaling attack} & {0.05 / 0.01}  & {0.05 / 0.01} & {0.06 / 0.02} & {0.06 / 0.03} & {0.07 / 0.05}  & {0.48 / 0.34}\\
		\hline
		{Adaptive attack} & {0.05}  & {0.05} & {0.06} & {0.09} & {0.48}  & {0.48}\\
		\hline
	    \end{tabular}%
        \label{har-bias}
    }}
    \xc{
    \subfloat[CH-MNIST]{
        \begin{tabular}{|c|c|c|c|c|c|c|}
		\hline
		{Bias probability} & {0.125}  & {0.2} & {0.4} & {0.6} & {0.8}  & {1.0}\\
		\hline
		{No attack} & {0.10}  & {0.10} & {0.11} & {0.13} & {0.13}  & {0.89}\\
		\hline
		{LF attack} & {0.12}  & {0.12} & {0.12} & {0.17} & {0.21}  & {0.89}\\
		\hline
		{Krum attack} & {0.12}  & {0.12} & {0.14} & {0.17} & {0.19}  & {0.89}\\
		\hline
		{Trim attack} & {0.13}  & {0.13} & {0.14} & {0.20} & {0.20}  & {0.89}\\
		\hline
		{Scaling attack} & {0.14 / 0.03}  & {0.14 / 0.02} & {0.15 / 0.02} & {0.16 / 0.06} & {0.14 / 0.01}  & {0.89 / 0.01}\\
		\hline
		{Adaptive attack} & {0.13}  & {0.14} & {0.14} & {0.89} & {0.89}  & {0.89}\\
		\hline
	    \end{tabular}%
        \label{har-bias}
    }}
	\label{tab:sample_bias}
	\vspace{-2mm}
\end{table}%

\myparatight{Comparing different variants of FLTrust} FLTrust has three key features: a root dataset, using ReLU to clip the cosine similarity scores, and normalizing each local model update. Depending on how each feature is used, we consider the following five variants of FLTrust:
\begin{itemize}
    \item {\bf FLTrust-Server.} In this variant, the server only uses the root dataset to train the global model. Therefore, there is no communications between the clients and the server during the training process. \xc{We use this variant to show that the server cannot obtain a good model using its root dataset alone. In other words, even if some clients are malicious, communicating with clients still improves the global model. } 
    \item {\bf FLTrust-withServer.} In this variant, the server computes the weighted average of the clients' local model updates together with the server model update whose trust score is 1.   
    \item {\bf FLTrust-NoReLU.} In this variant, the server does not use ReLU to clip the cosine similarity scores of the local model updates when computing their trust scores. 
    \item {\bf FLTrust-NoNorm.} In this variant, the server does not normalize the local model updates to have the same magnitude as the server model update. 
    \item {\bf FLTrust-ParNorm.} In this variant, the server applies partial normalization, i.e., only normalizes the local model updates whose magnitudes are larger than that of the server model update to  have the same magnitude as the server model update. 

\end{itemize}

\begin{figure}[!t] 
	\centering
	\vspace{-3mm}		
	\subfloat[Testing error rate]{\includegraphics[width=0.24 \textwidth]{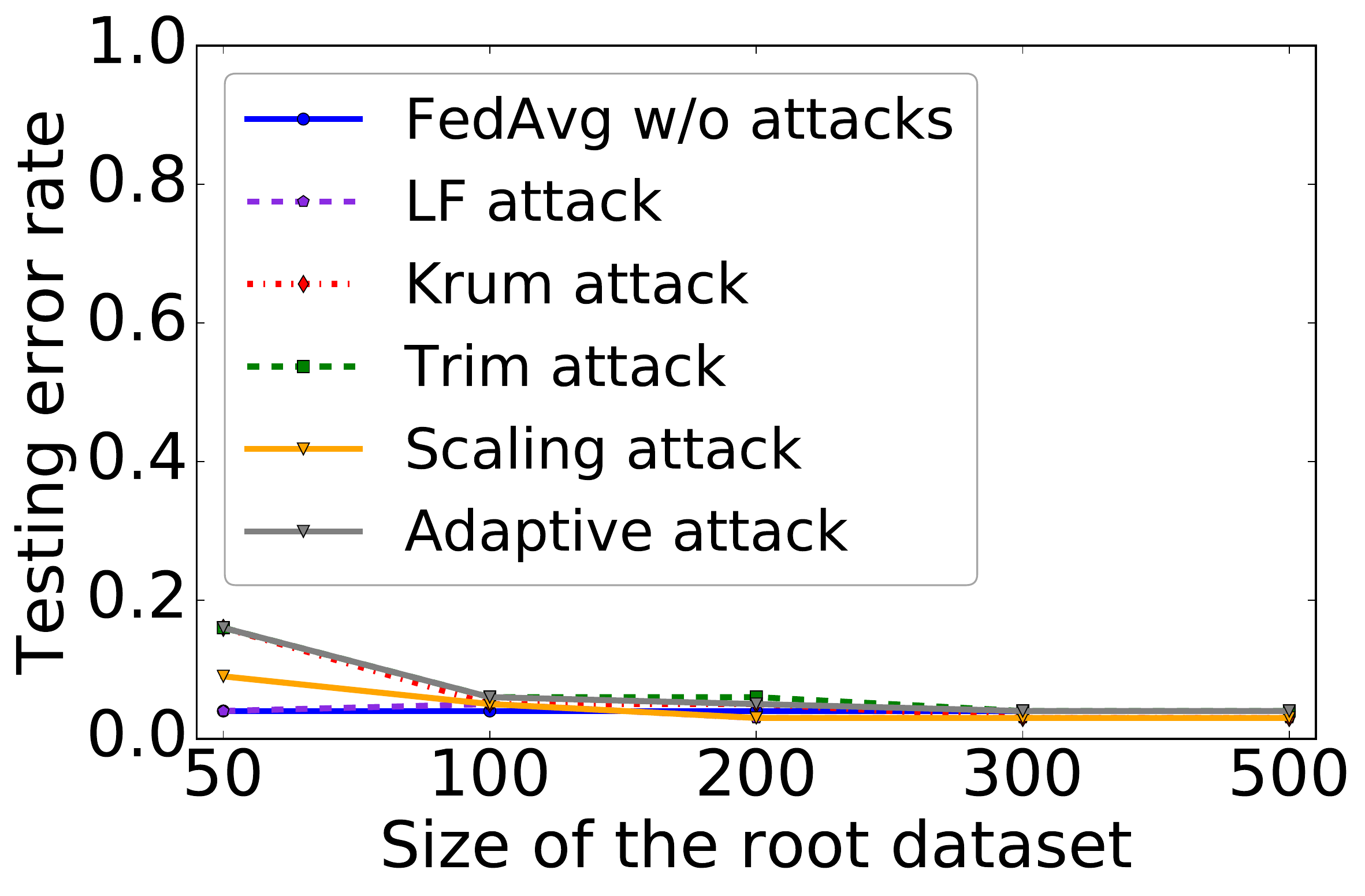}\label{fig:untargeted_rootsize}}
	\subfloat[Attack success rate]{\includegraphics[width=0.24 \textwidth]{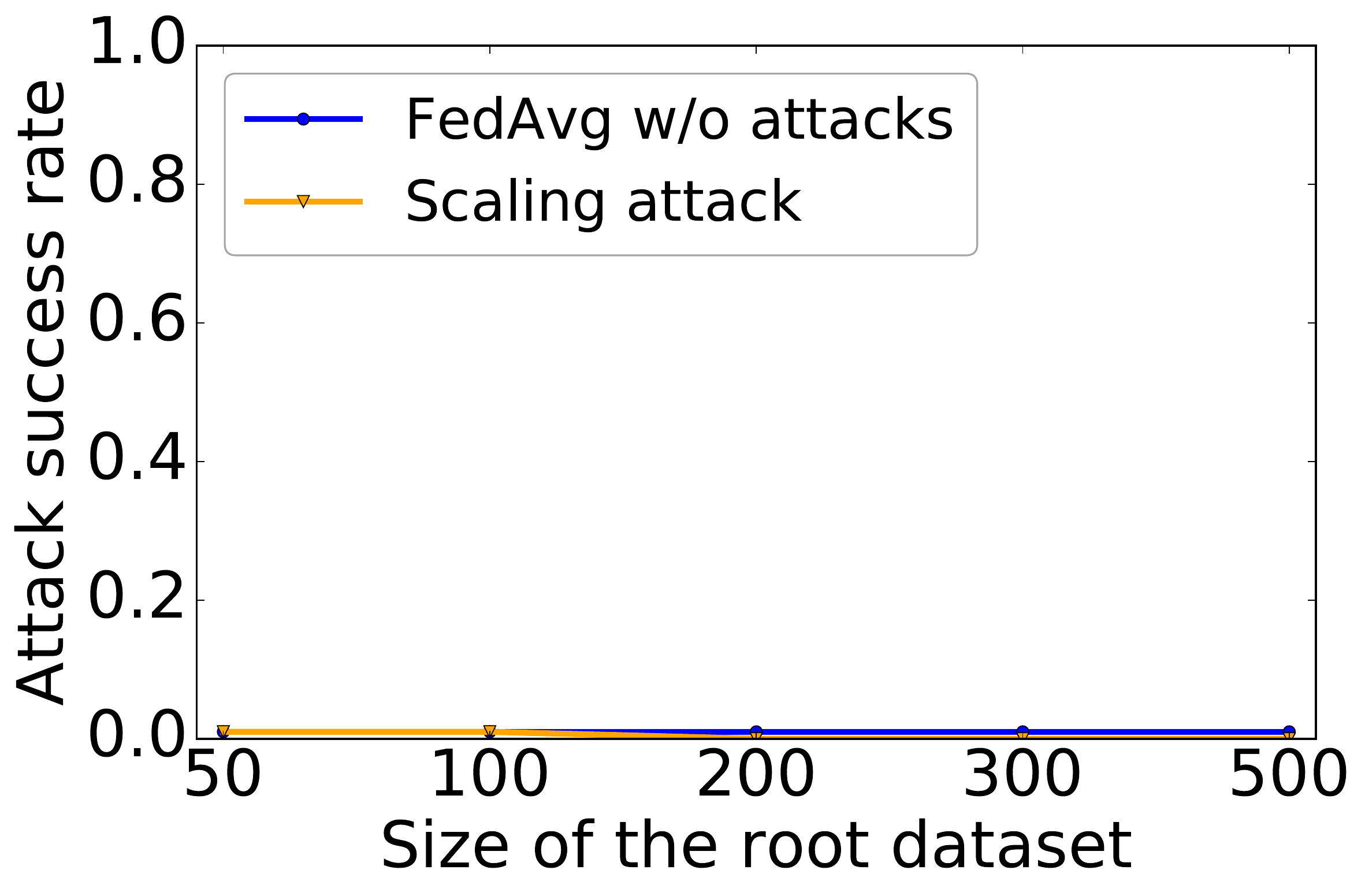}\label{fig:scaling_rootsize}}
	\caption{\xc{Impact of the root dataset size on FLTrust under different attacks for MNIST-0.5.}} 
	\label{fig:num_sample}
	\vspace{-3mm}
\end{figure}

Table \ref{tab:variants} compares the variants with respect to their testing error rates under different attacks and the attack success rates of the Scaling attacks on MNIST-0.5. The attacks are not applicable to FLTrust-Server as it does not require communications from the clients. Our results show that FLTrust outperforms the five variants. FLTrust outperforms FLTrust-Server and FLTrust-withServer because the root dataset is small. The fact that FLTrust outperforms 
FLTrust-NoReLU, FLTrust-NoNorm, and FLTrust-ParNorm indicates the necessity of our ReLU operation and normalization.

\myparatight{Impact of the root dataset} Our root dataset can be characterized by its size and how it is sampled (i.e., Case I vs. Case II). Therefore, we study the impact of the root dataset on FLTrust with respect to its size and how it is sampled. Figure \ref{fig:num_sample} shows the testing error rates of FLTrust under different attacks and the attack success rates under the Scaling attack on MNIST-0.5 when the size of the root dataset increases from 50 to 500, where the root dataset is sampled uniformly in Case I. We observe that a root dataset with only 100 training examples is sufficient for FLTrust to defend against the attacks. Specifically, when  the root dataset has 100 training examples, the testing error rates of FLTrust under  attacks are similar to that of FedAvg without attacks, and the attack success rate of the Scaling attack is close to 0. When the size of the root dataset increases beyond 100, the testing error rates and attack success rates of FLTrust  further decrease slightly. 

\xc{We also evaluate the impact of the bias probability in Case II. Table \ref{tab:sample_bias} shows the testing error rates of FLTrust under different attacks and the attack success rates of the Scaling attacks when the bias probability varies. The second column in each table corresponds to the bias probability with which Case II reduces to Case I. We  increase the bias probability up to 1.0 to simulate larger difference between the root data distribution and the overall training data distribution. We observe that FLTrust is accurate and robust when the bias probability is not too large. For instance, when the bias probability is no more than 0.4 for MNIST-0.5, the testing error rates of FLTrust under attacks are at most 0.08, compared to 0.05 when the bias probability is 0.1. Our results show that FLTrust works well when the root data distribution does not diverge too much from the overall training data distribution.} 

\begin{figure}[!t]
	\centering
	\subfloat[LF attack]{\includegraphics[width=0.24 \textwidth]{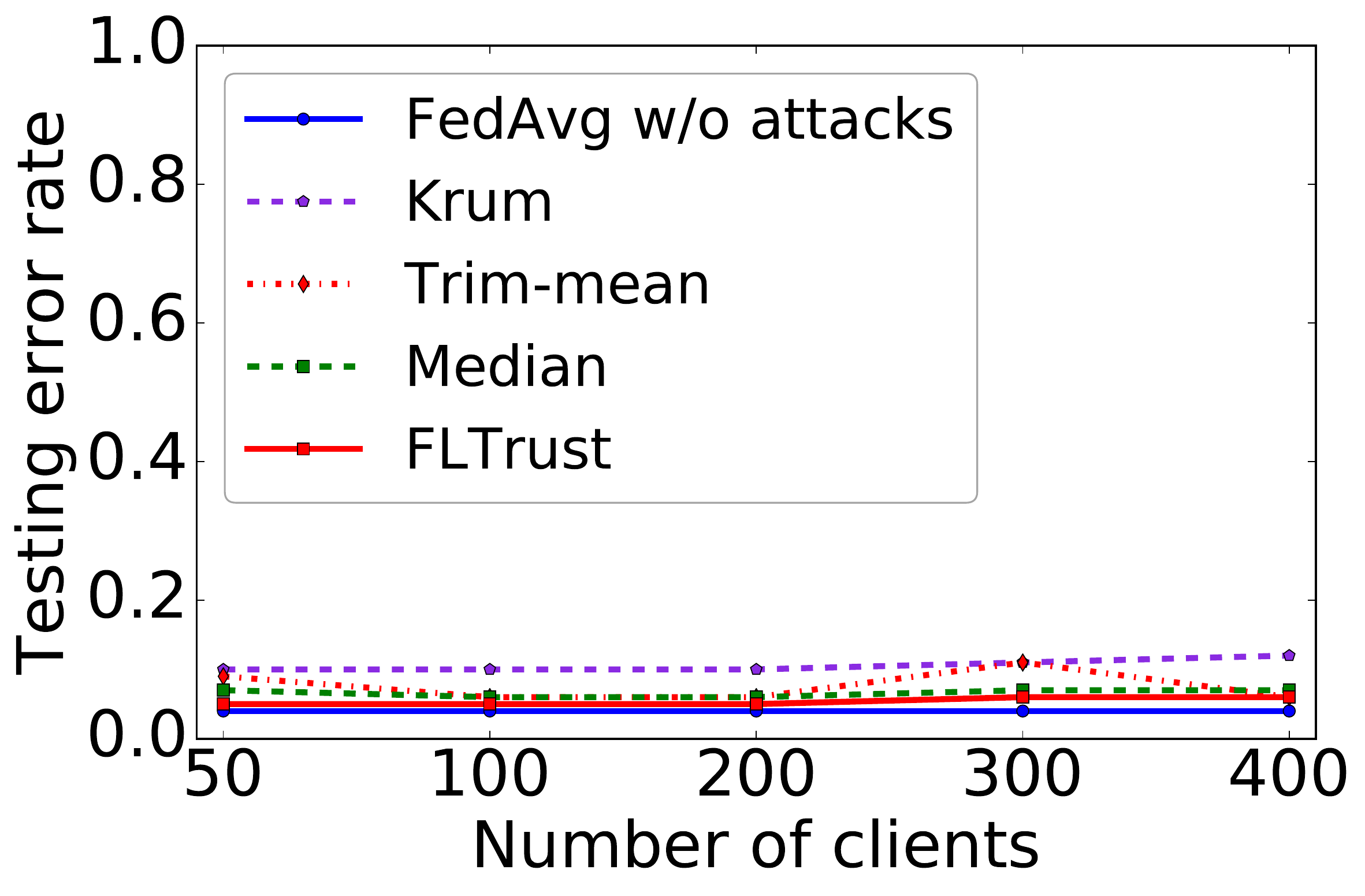}}
	\subfloat[Krum attack]{\includegraphics[width=0.24 \textwidth]{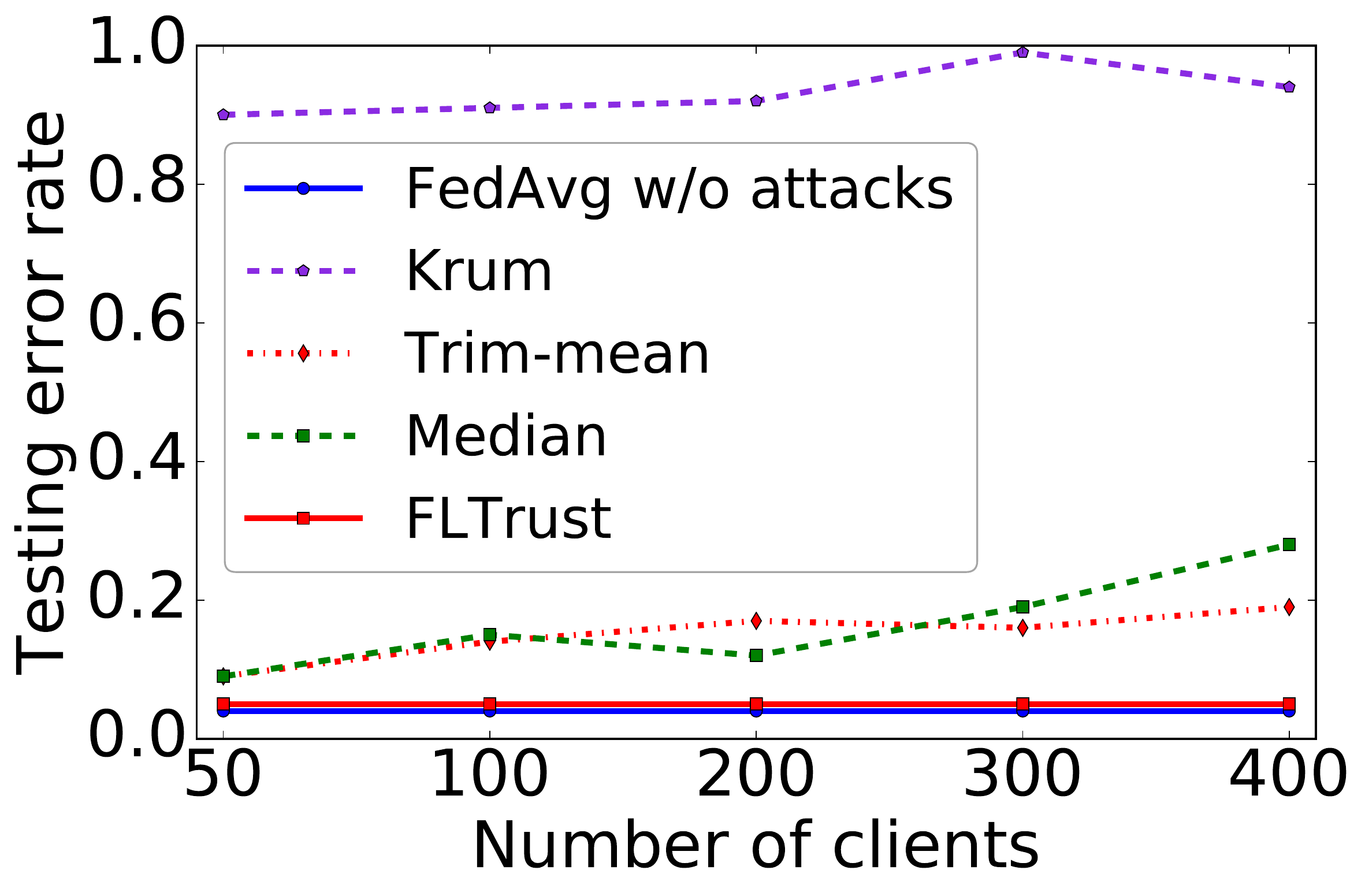}}\\
	\vspace{-1mm}
	\subfloat[Trim attack]{\includegraphics[width=0.24 \textwidth]{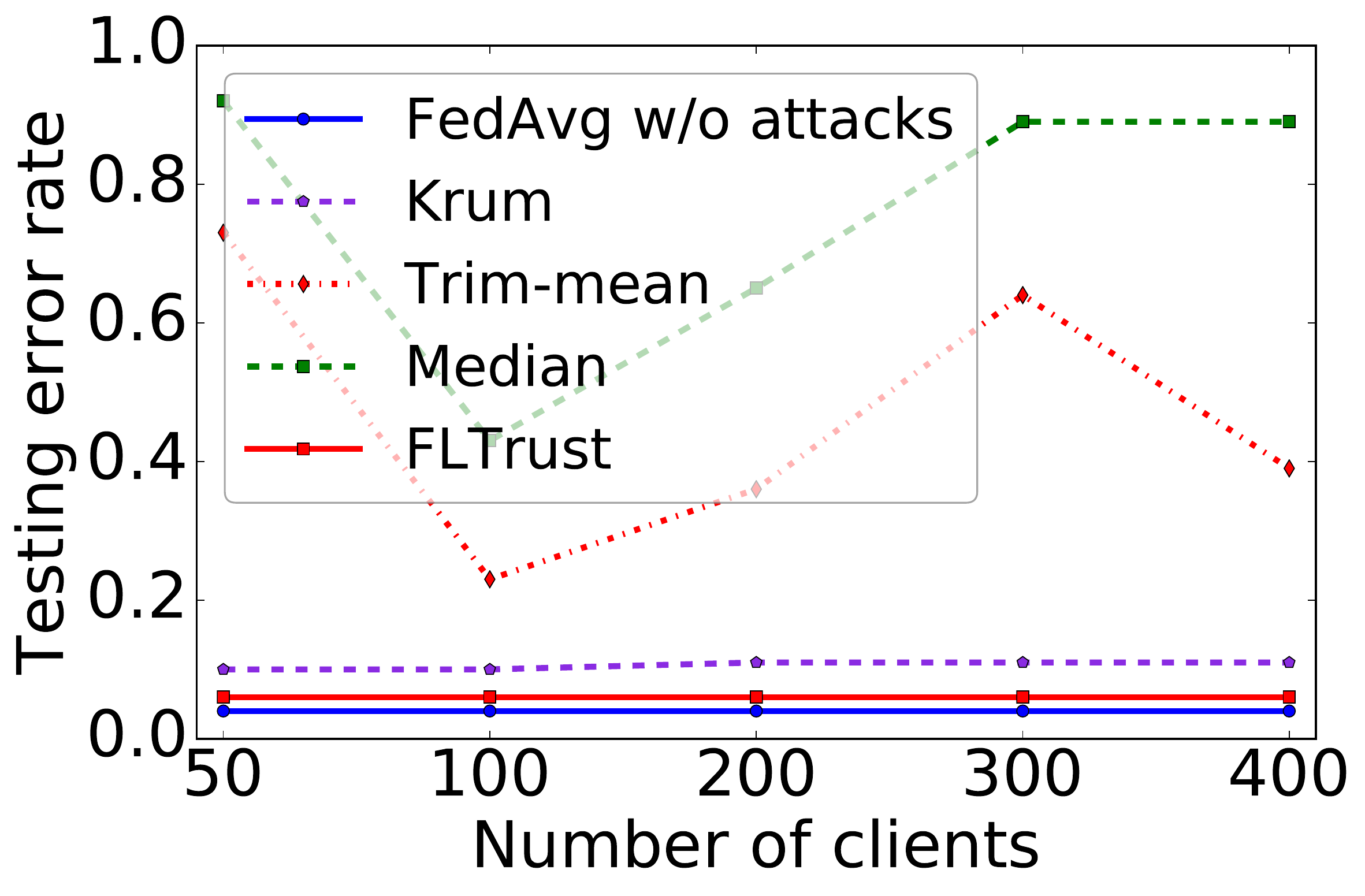}} 
	\subfloat[Scaling attack]{\includegraphics[width=0.24 \textwidth]{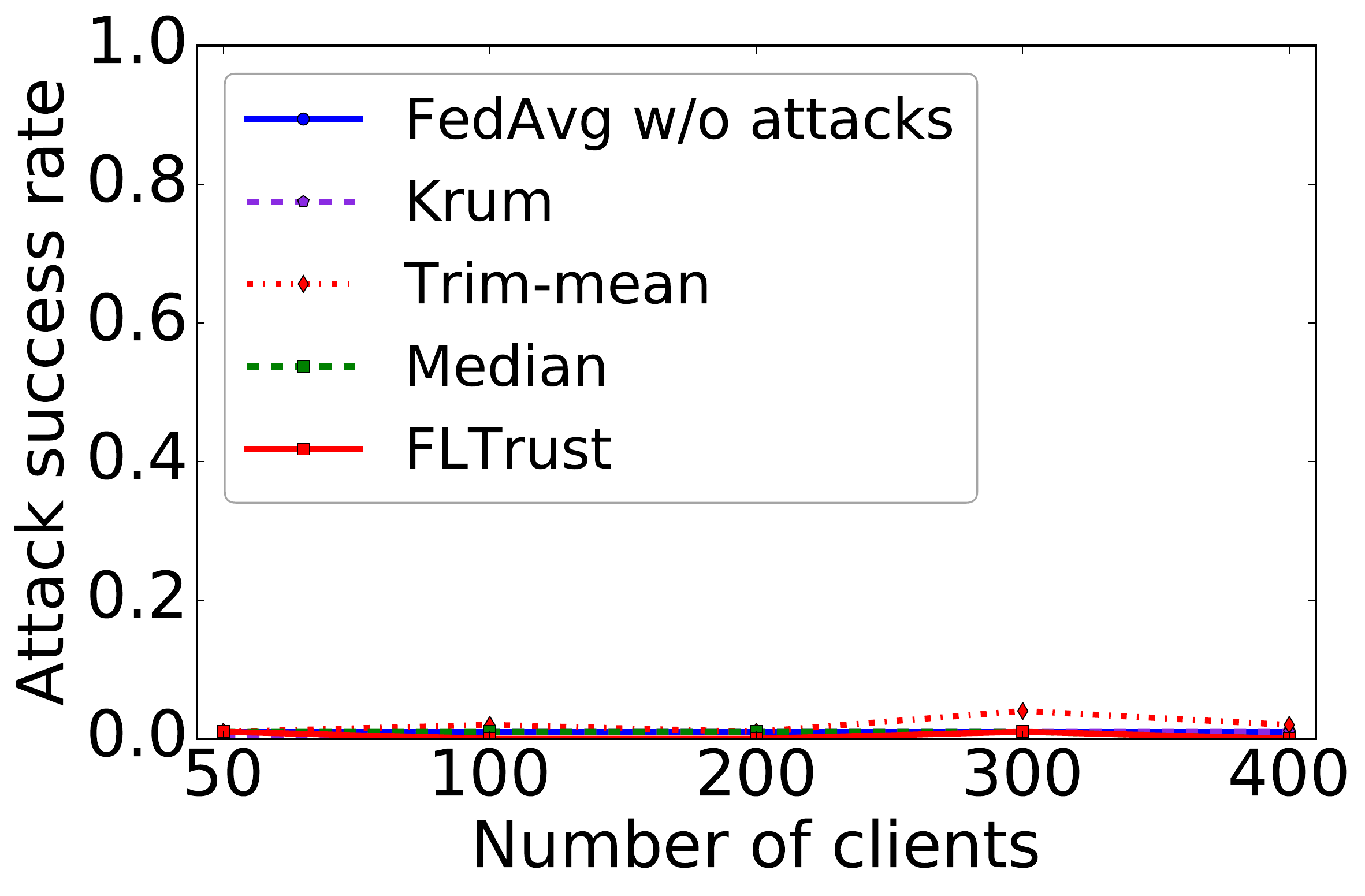}}\\
	\vspace{-1mm}
	\subfloat[\xc{Adaptive attack}]{\includegraphics[width=0.24 \textwidth]{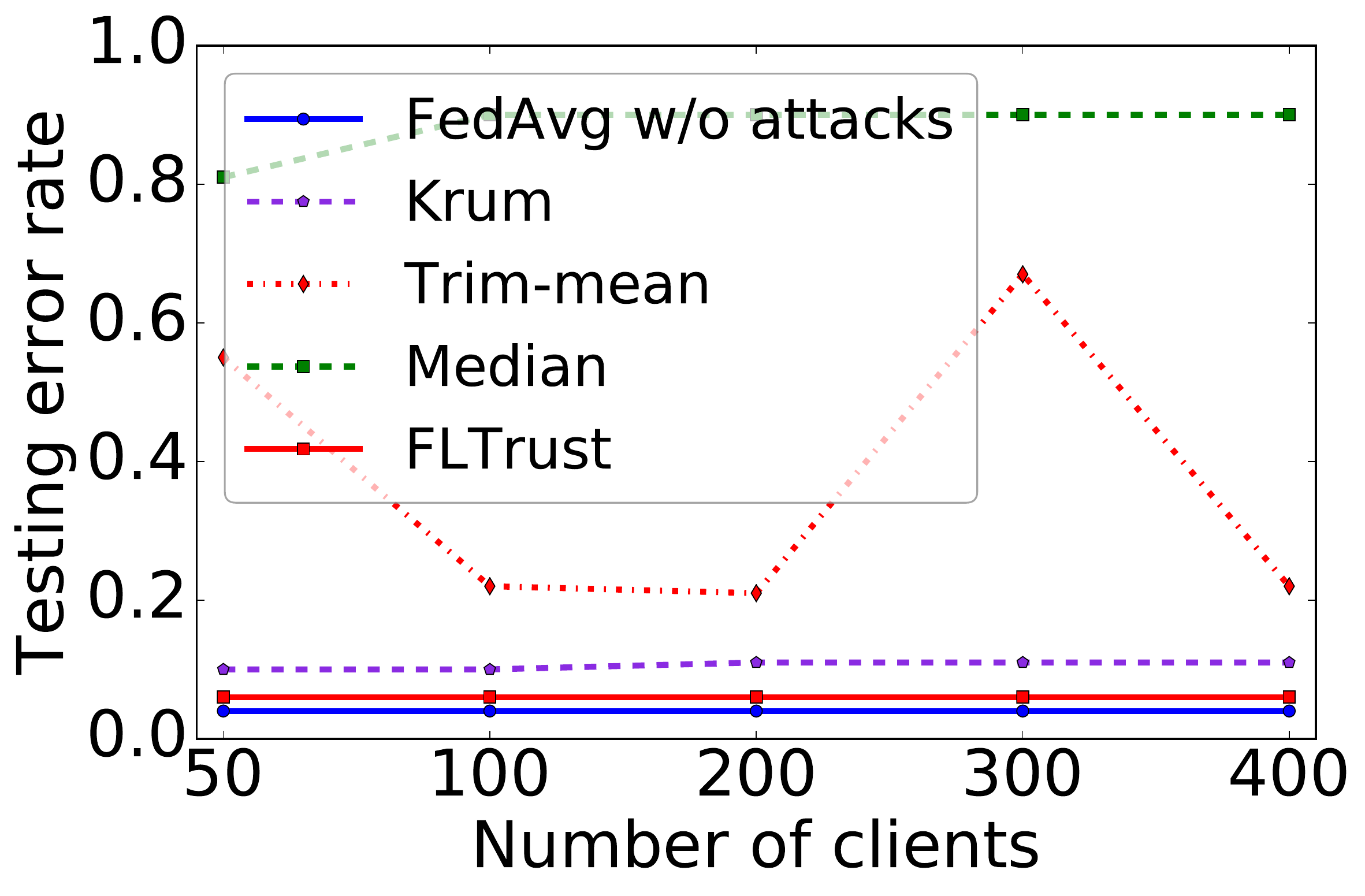}}
	\vspace{1mm}
	\caption{Impact of the total number of clients on the testing error rates of different FL methods under different attacks ((a)-(c)) and the attack success rates of the Scaling attacks, where MNIST-0.5 is used. The testing error rates of all the compared FL methods are similar and small under the Scaling attacks, which we omit for simplicity.} 
	\label{fig:num_clients}
	\vspace{-5mm}
\end{figure}

\begin{figure}[!t]
	\centering
	\subfloat[LF attack]{\includegraphics[width=0.24 \textwidth]{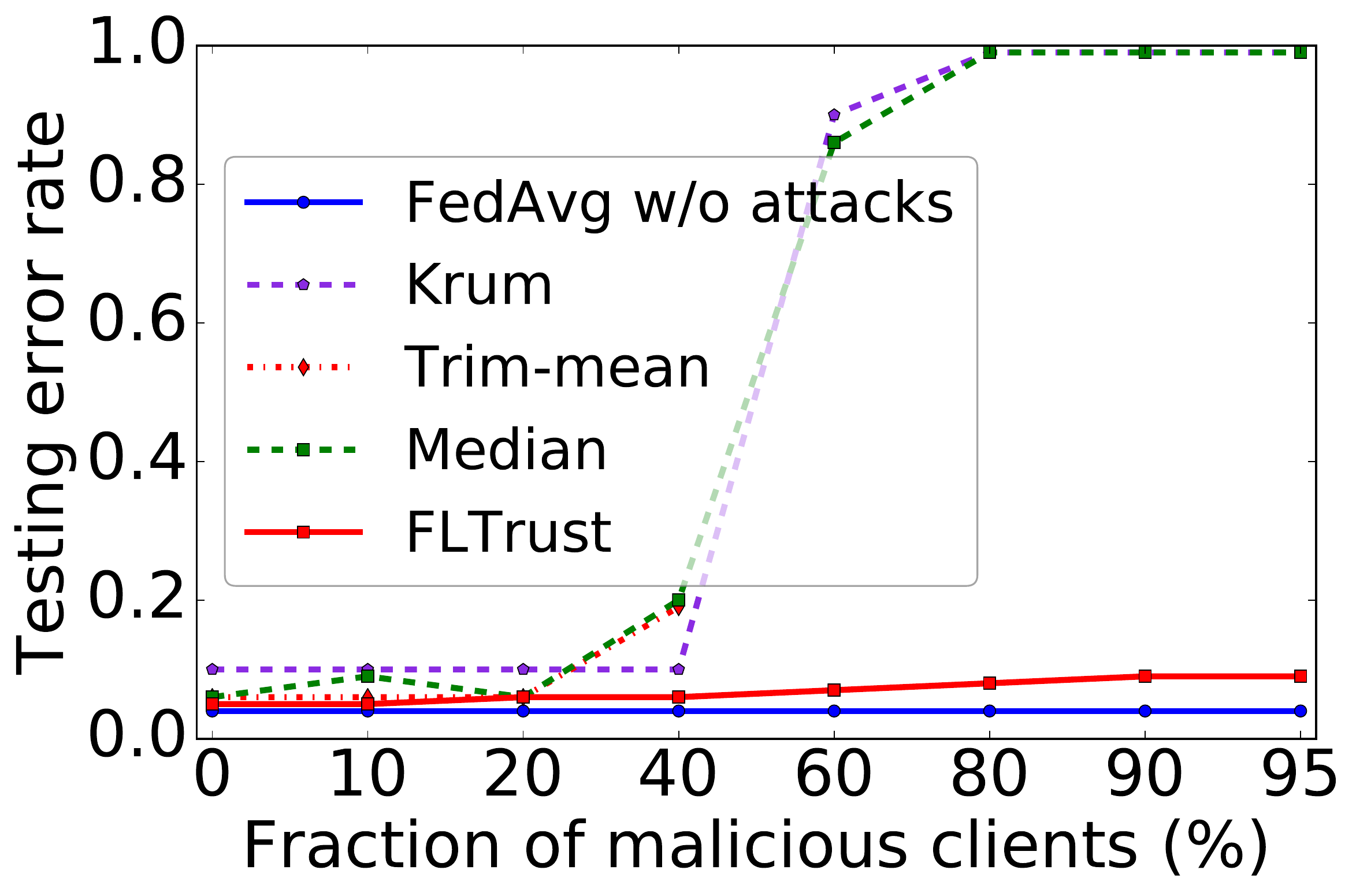}\label{fig:num_malicious_lf}}
	\subfloat[Krum attack]{\includegraphics[width=0.24 \textwidth]{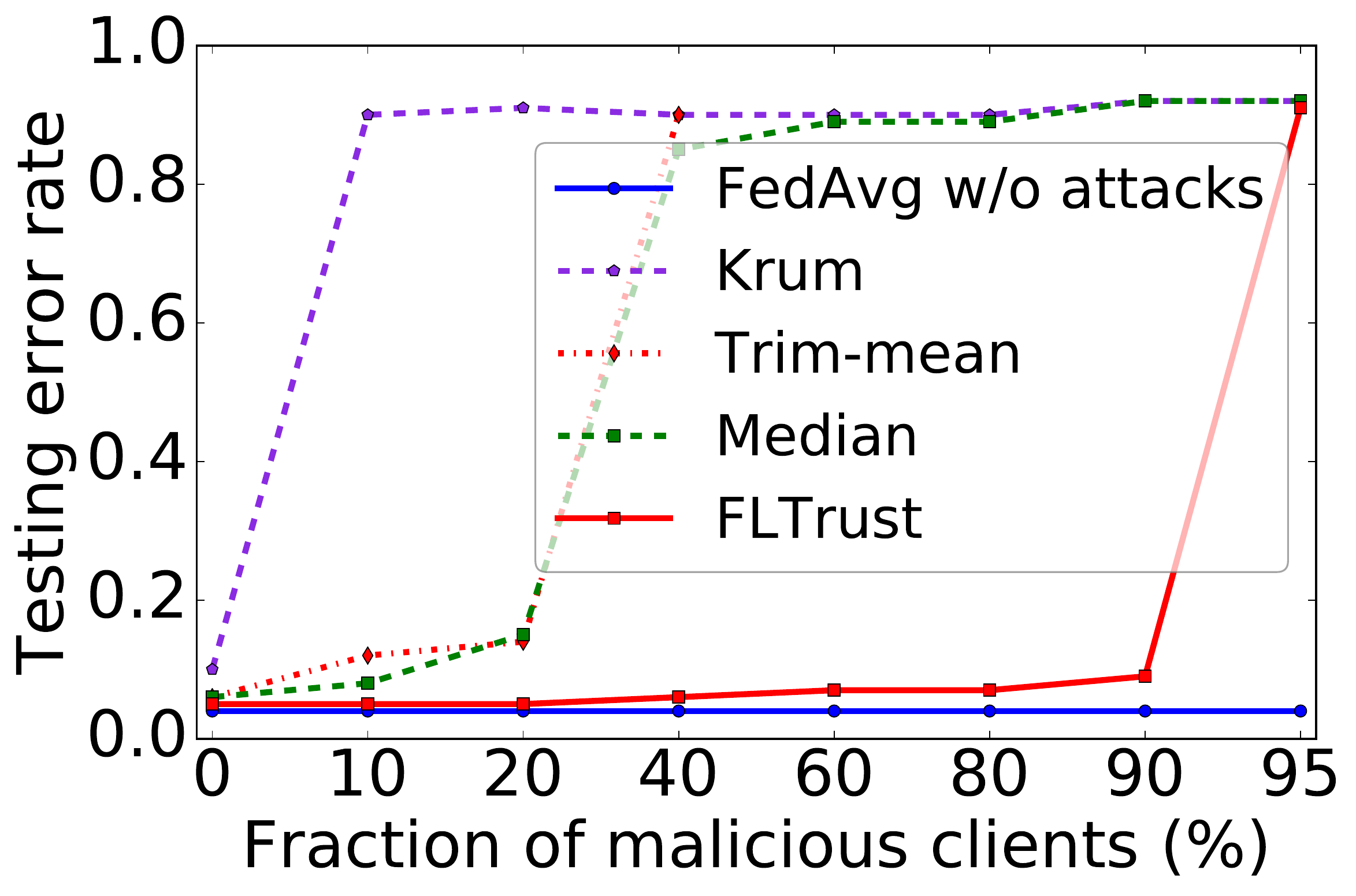}\label{fig:num_malicious_krum}}\\
	\vspace{-1mm}
	\subfloat[Trim attack]{\includegraphics[width=0.24 \textwidth]{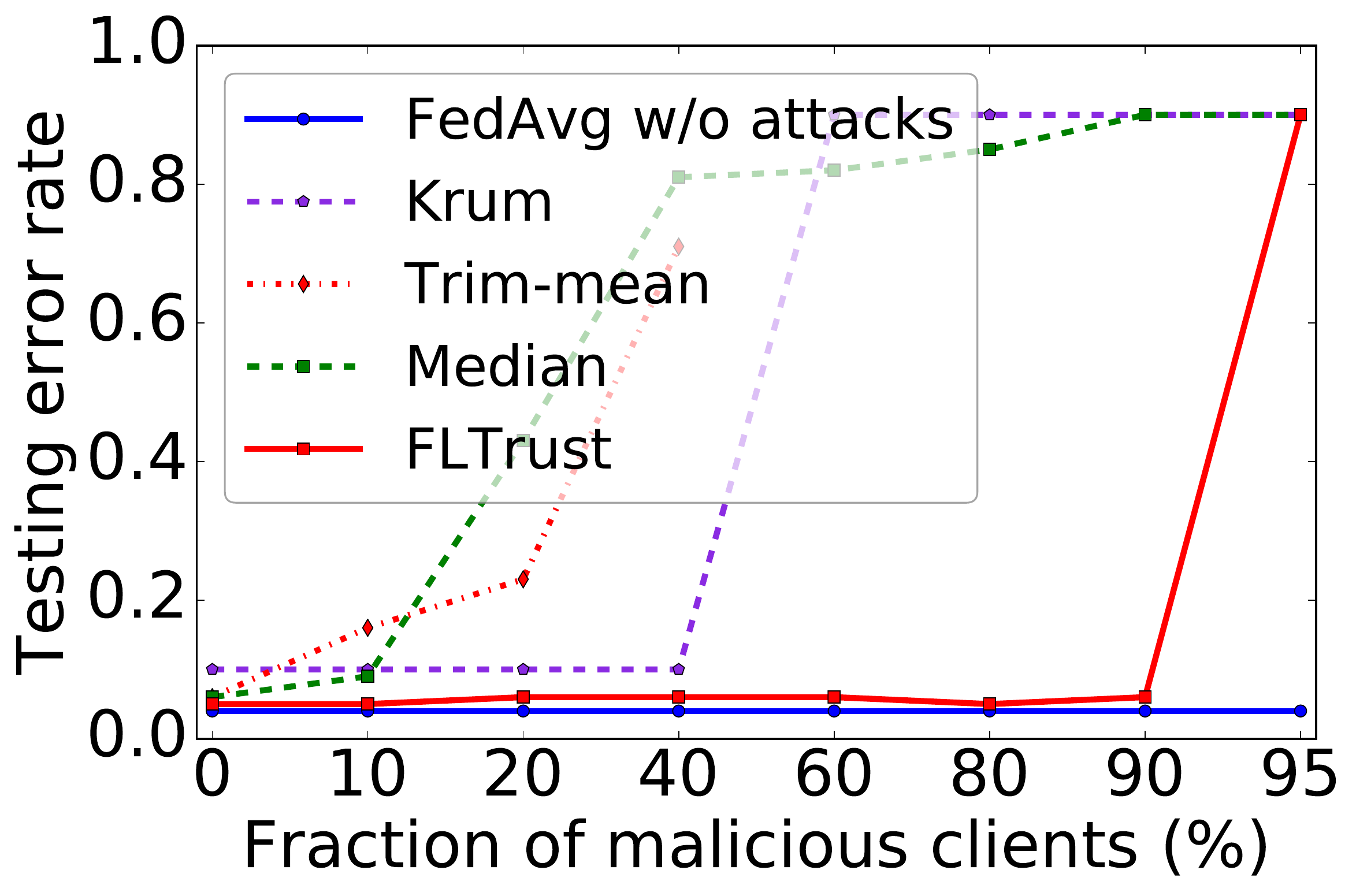}\label{fig:num_malicious_trim}} 
	\subfloat[Scaling attack]{\includegraphics[width=0.24 \textwidth]{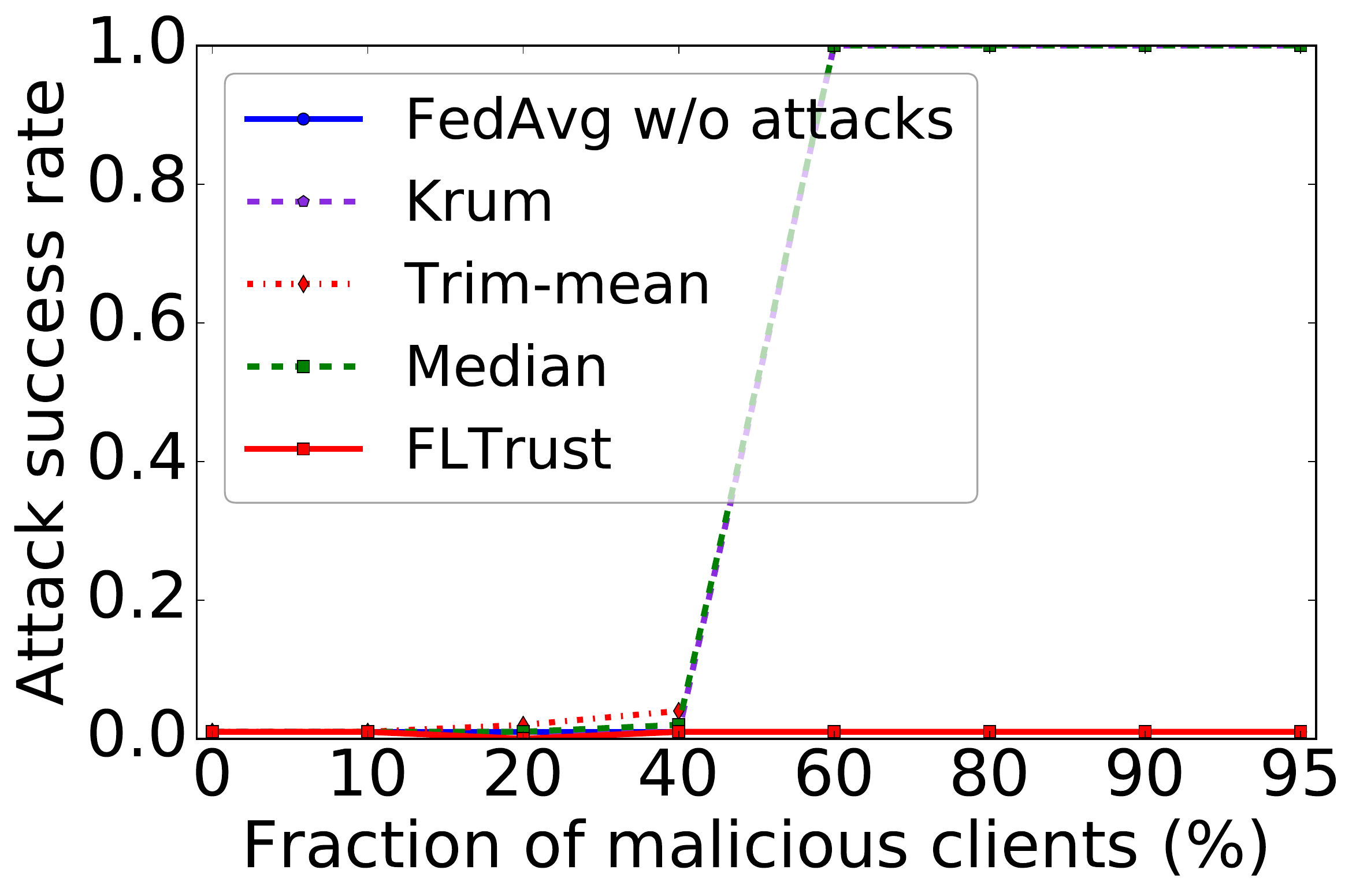}\label{fig:num_malicious_scale}}\\
	\vspace{-1mm}
	\subfloat[\xc{Adaptive attack}]{\includegraphics[width=0.24 \textwidth]{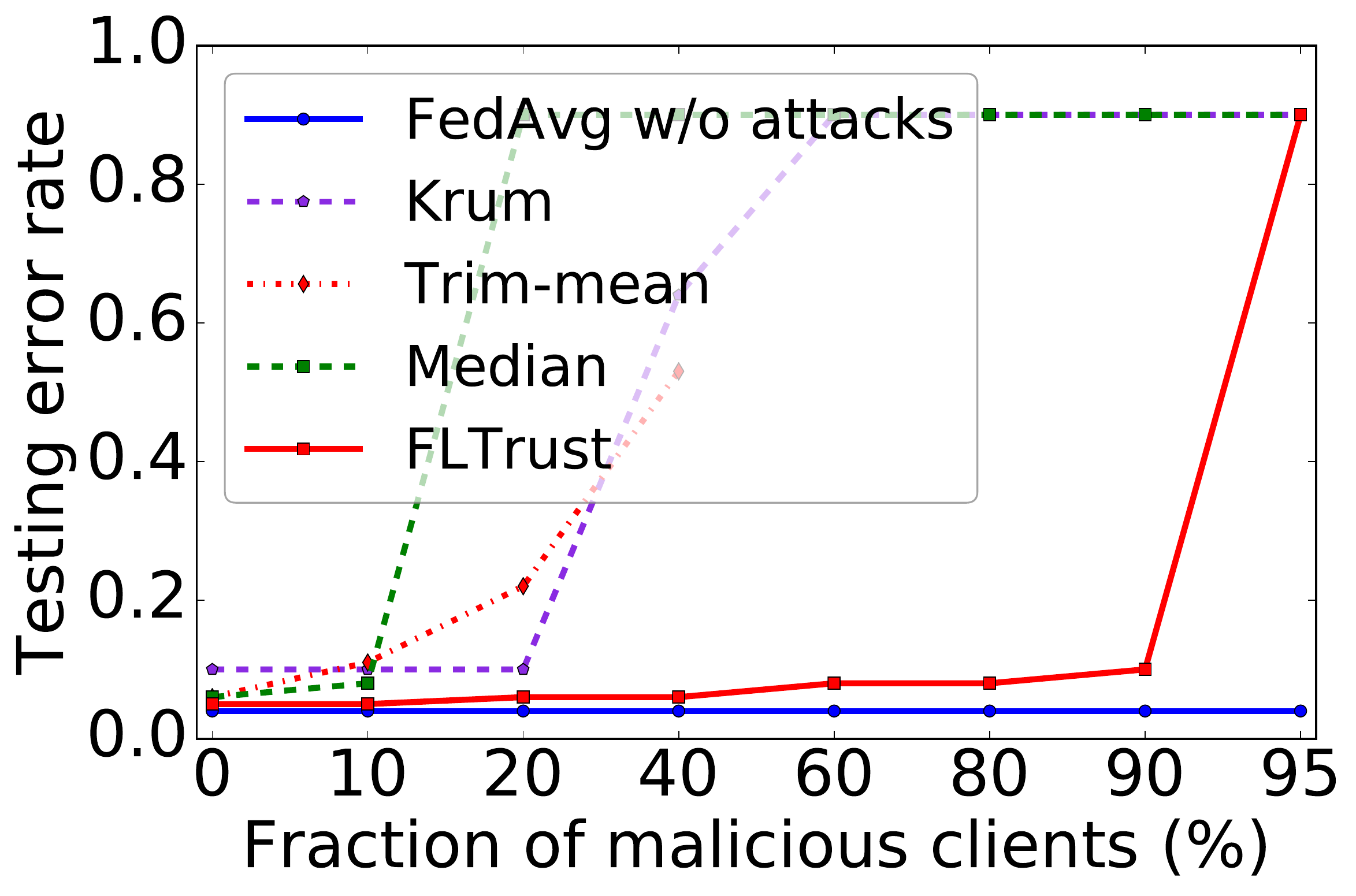}\label{fig:num_malicious_scale}}
	\vspace{1mm}
	\caption{Impact of the fraction of malicious clients on the testing error rates of different FL methods under different attacks ((a)-(c)) and the attack success rates of the Scaling attacks, where MNIST-0.5 is used. The testing error rates of all the compared FL methods are similar and small under the Scaling attacks, which we omit for simplicity.}
	\label{fig:num_malicious}
\end{figure} 

\begin{figure}[!t]
	\centering
	\subfloat[MNIST-0.1]{\includegraphics[width=0.24 \textwidth]{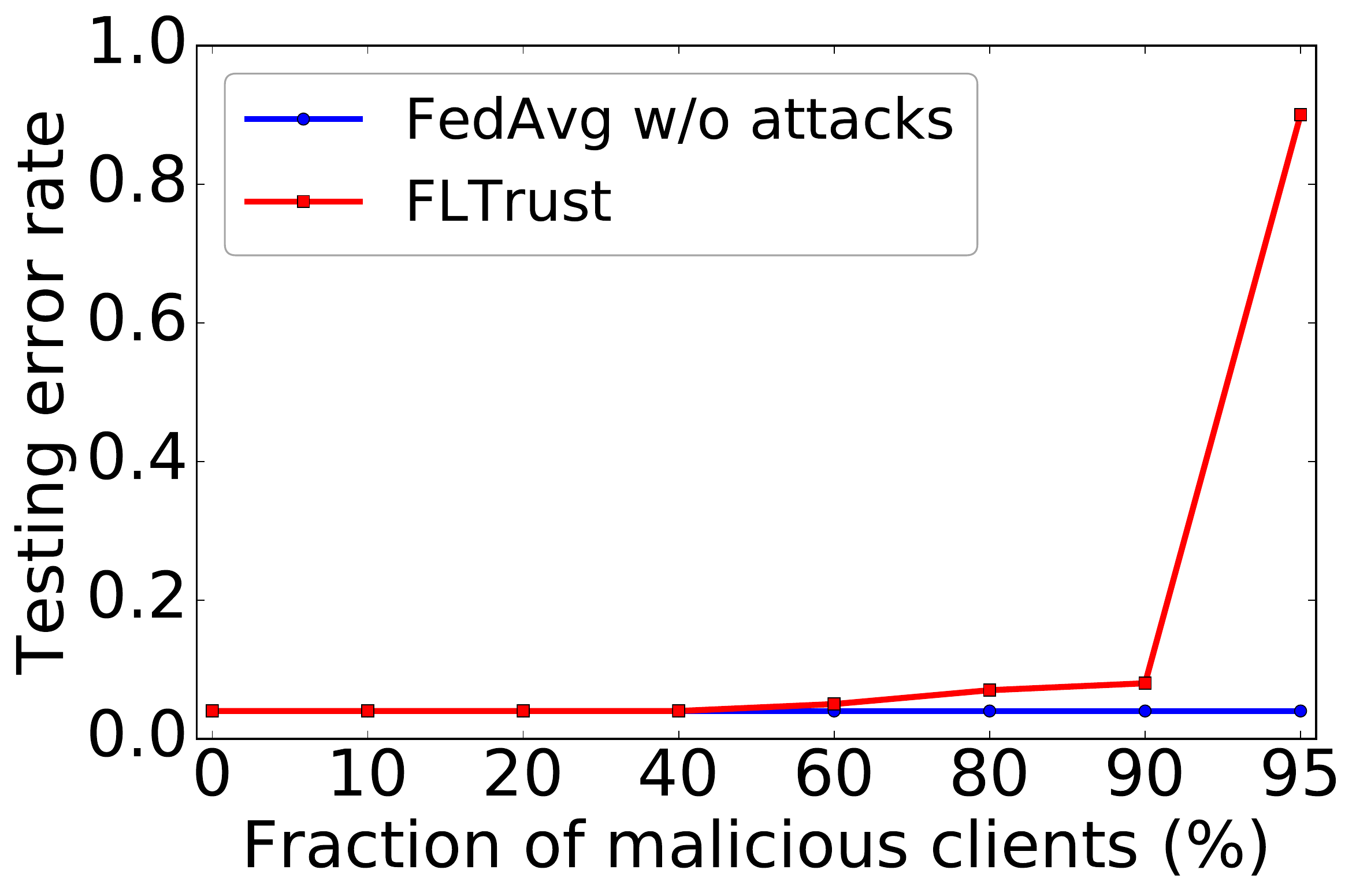}}
	\subfloat[MNIST-0.5]{\includegraphics[width=0.24 \textwidth]{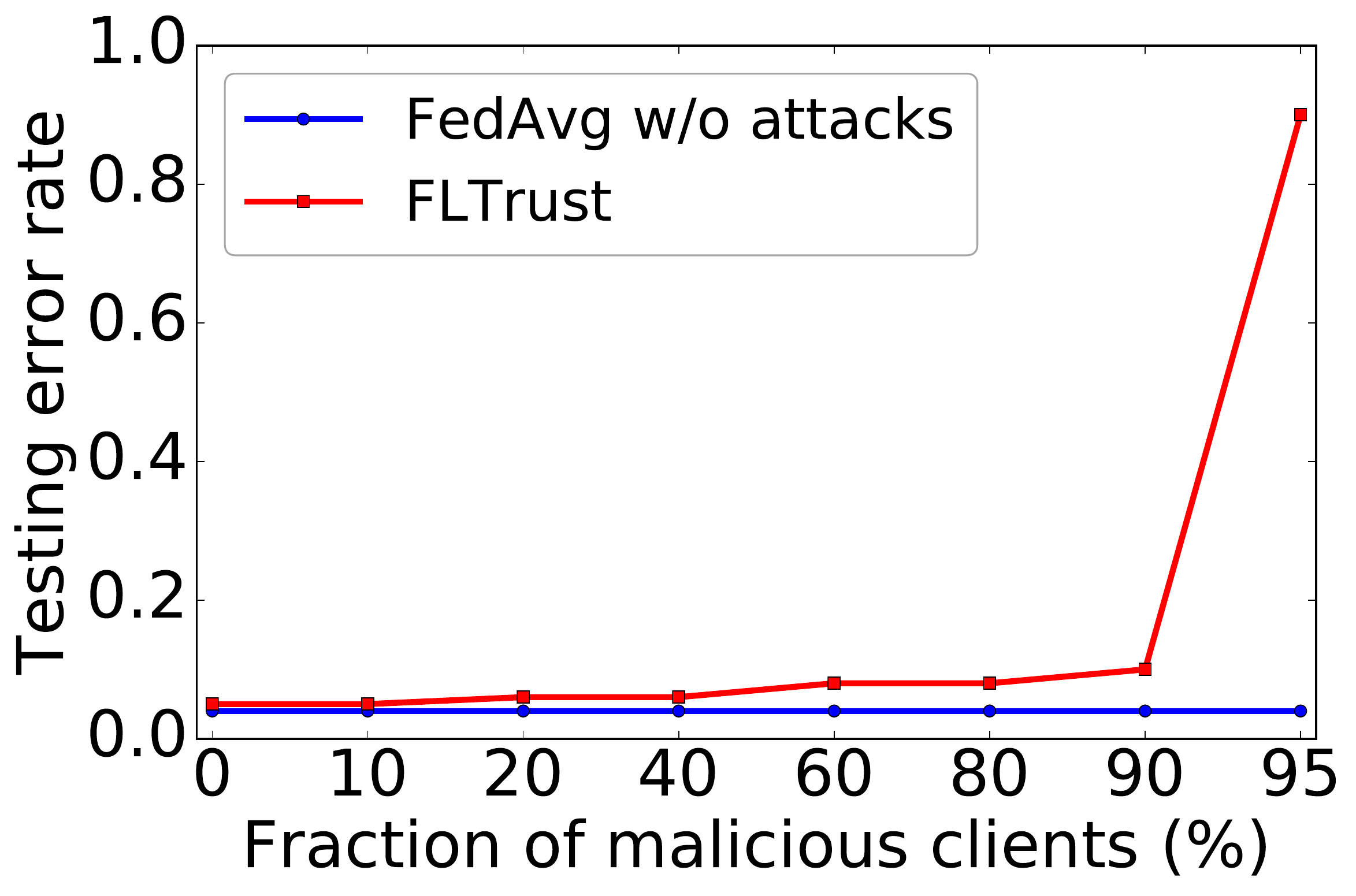}} \\
	\vspace{-1mm}
	\subfloat[Fashion-MNIST]{\includegraphics[width=0.24 \textwidth]{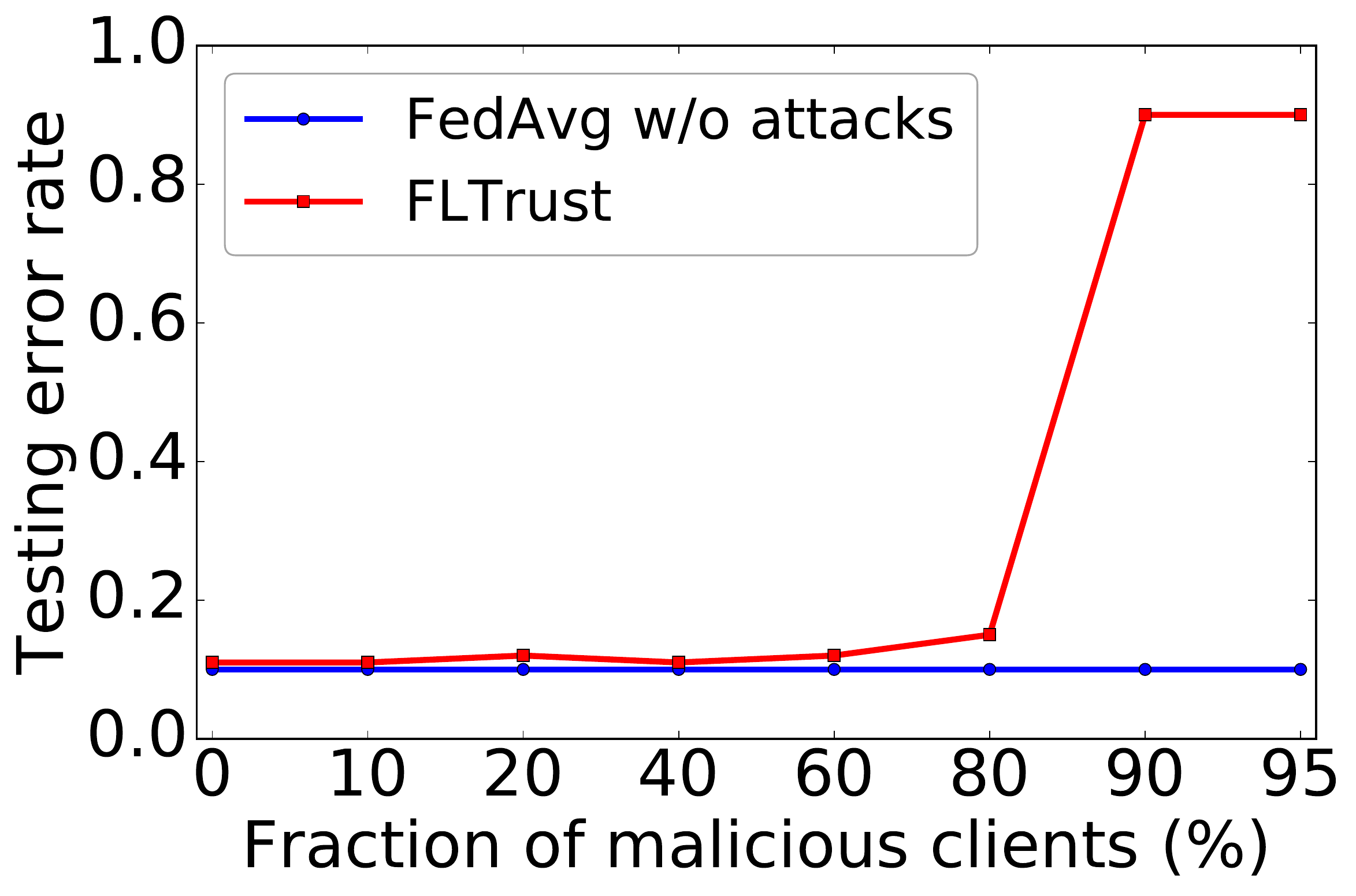}}
	\subfloat[CIFAR-10]{\includegraphics[width=0.24 \textwidth]{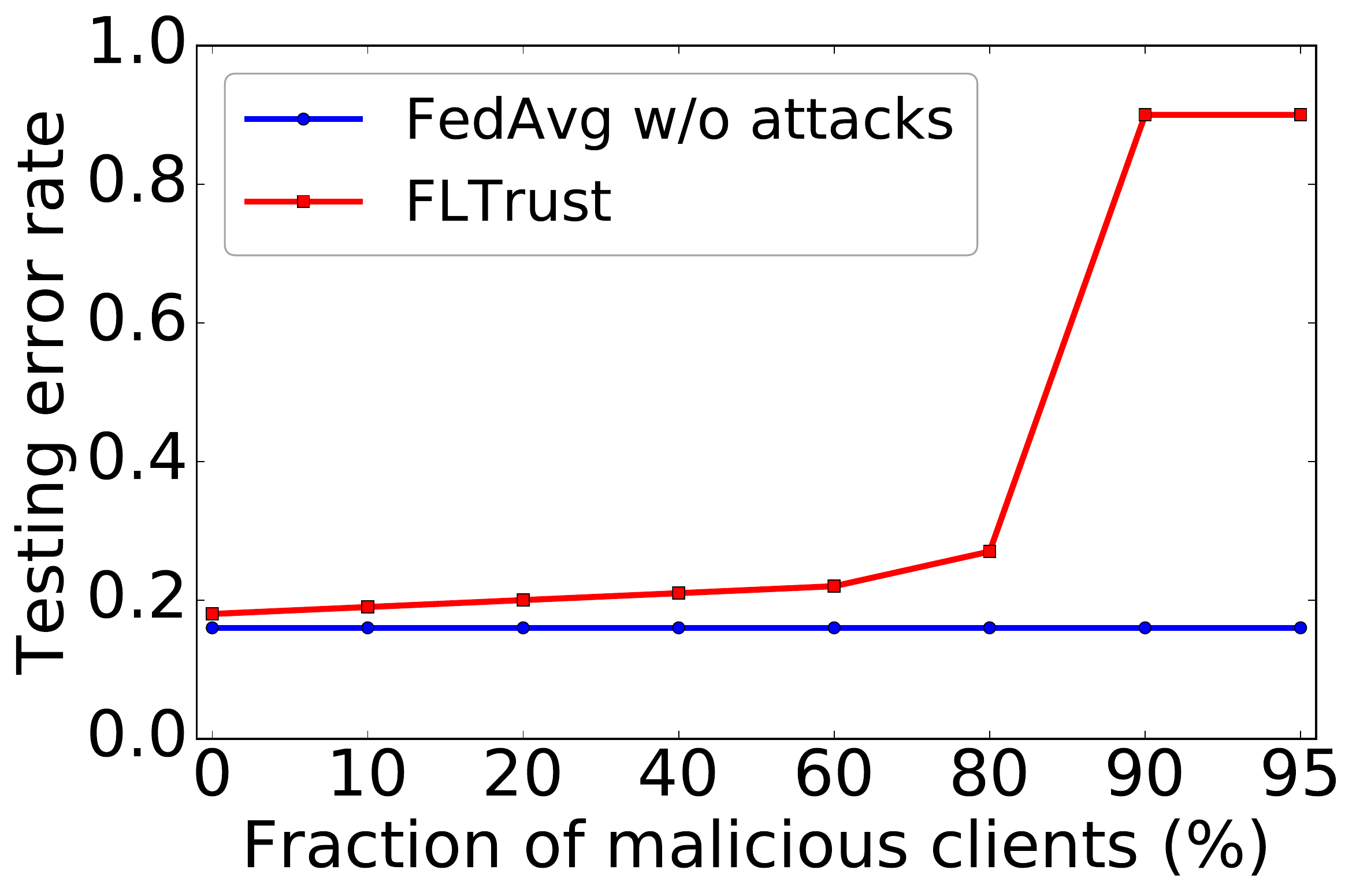}}\\
	\vspace{-1mm}
	\subfloat[HAR]{\includegraphics[width=0.24 \textwidth]{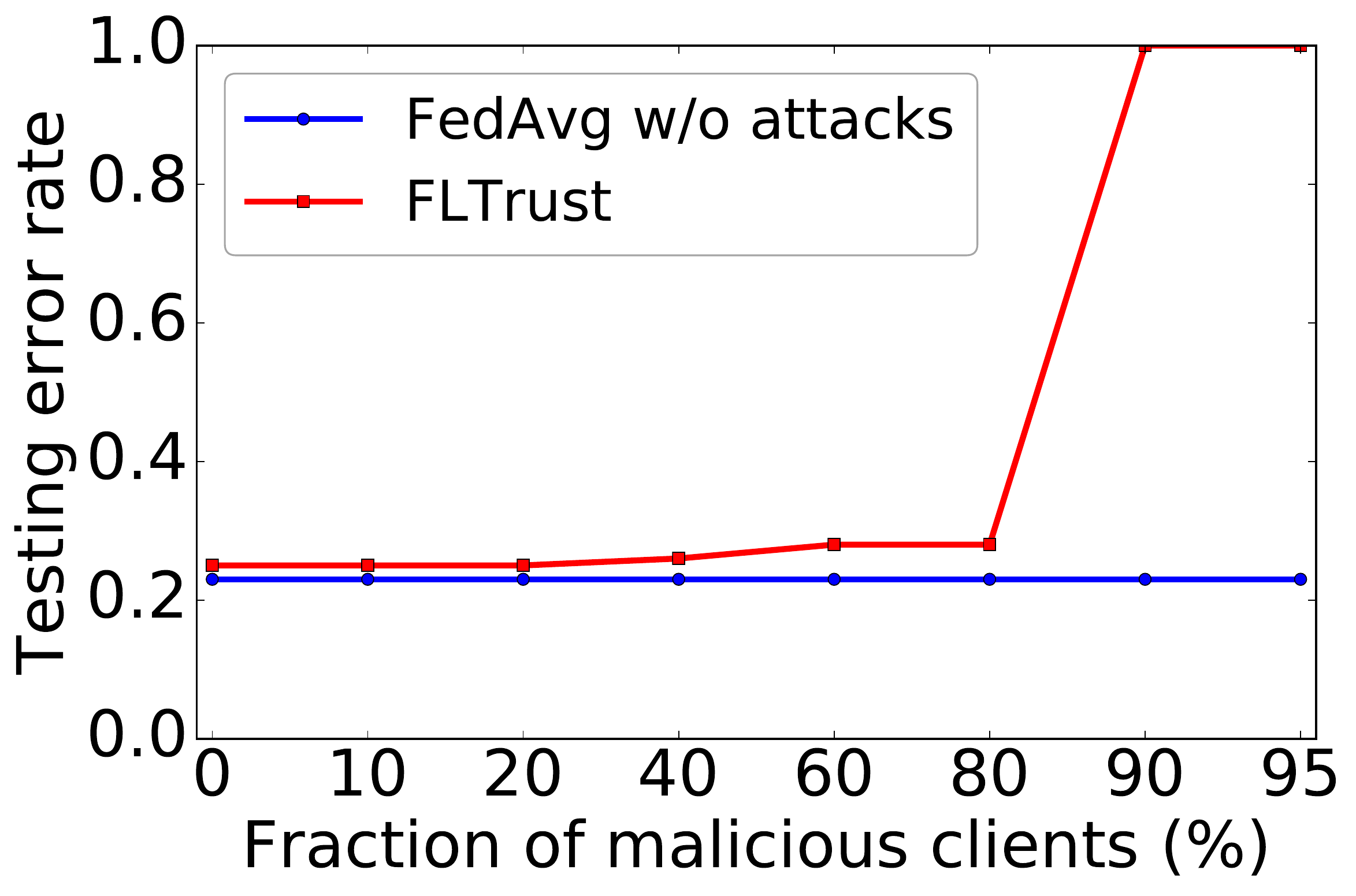}}
	\subfloat[\xc{CH-MNIST}]{\includegraphics[width=0.24 \textwidth]{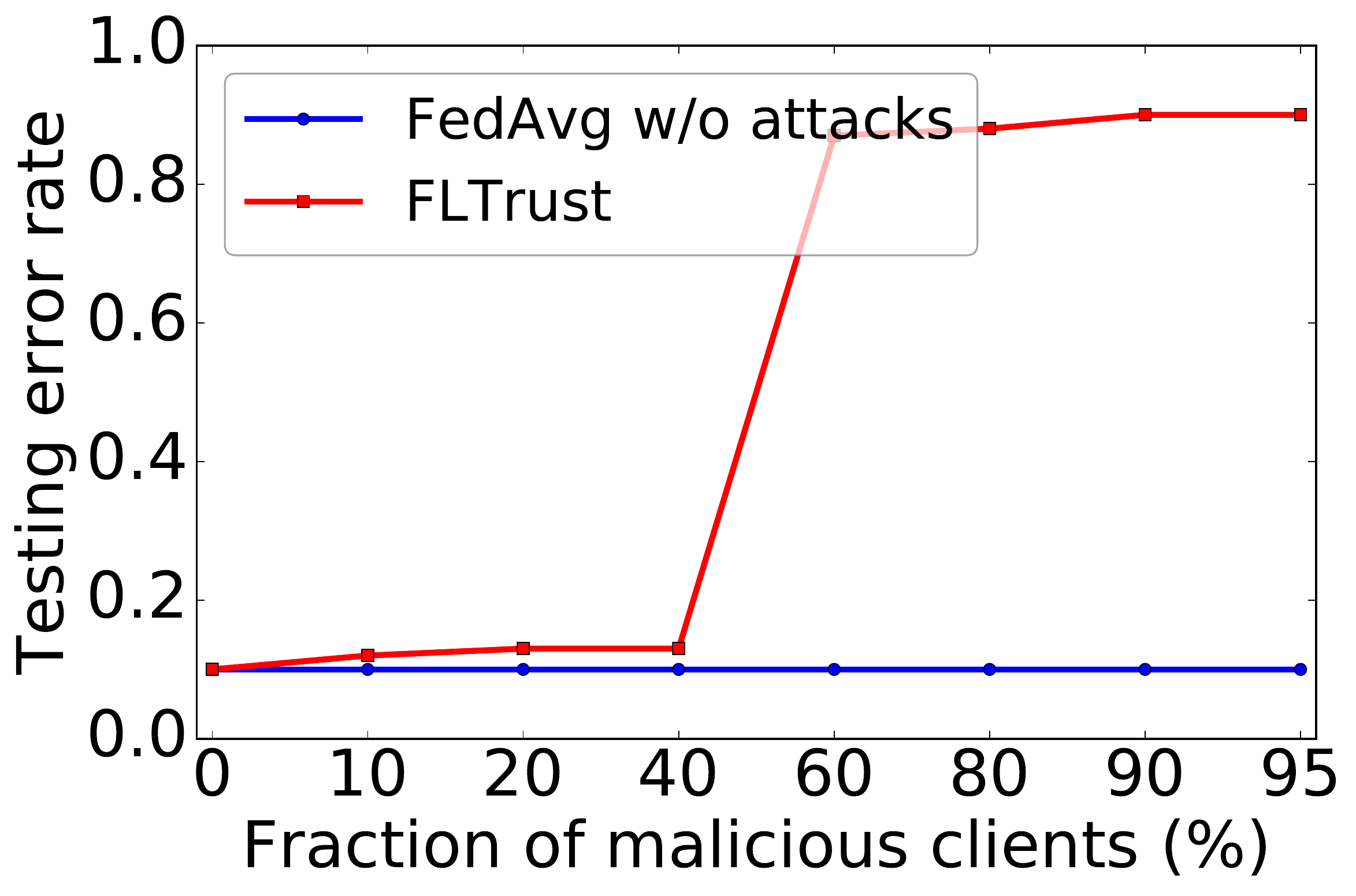}}
	\vspace{1mm}
	\caption{Impact of the fraction of malicious clients on the testing error rates of FLTrust under the adaptive attacks.}
	\label{adpative_attack}
\end{figure}

\myparatight{Impact of the total number of clients}
Figure \ref{fig:num_clients} shows the testing error rates of different FL methods under different attacks, as well as the attack success rates of the Scaling attacks, when the total number of clients $n$ increases from 50 to 400. We set the fraction of malicious clients to be $\frac{m}{n}=20\%$. We observe that FLTrust can defend against the attacks for all considered total number of clients. Specifically, FLTrust under attacks achieves testing error rates similar to FedAvg under no attacks, while the attack success rates of the Scaling attacks are close to 0 for FLTrust. Existing  methods can defend against the Scaling attacks on MNIST-0.5, i.e., the attack success rates are close to 0.  However, they cannot defend against the Krum attack, Trim attack, \xc{and/or our adaptive attack}, i.e., their corresponding testing error rates are large.  

\myparatight{Impact of the number of malicious clients} Figure \ref{fig:num_malicious} shows the testing error rates of different FL methods under different attacks and the attack success rates of the Scaling attacks on MNIST-0.5, when the fraction of malicious clients increases from 0 to 95\%.  Trim-mean cannot be applied when the fraction of malicious clients exceeds 50\% because the number of local model updates removed by Trim-mean is twice of the number of malicious clients. Therefore, for Trim-mean, we only show the results when the malicious clients are less than 50\%. 

We observe that, under existing attacks \xc{and our adaptive attacks}, FLTrust can tolerate up to 90\% of malicious clients. Specifically, FLTrust under these attacks still achieves testing error rates similar to FedAvg without attacks when up to 90\% of the clients are malicious, while the attack success rates of the Scaling attacks for FLTrust are still close to 0 when up to 95\% of the clients are malicious. However, existing Byzantine-robust FL methods can tolerate much less malicious clients.  
For instance, under Krum attack, the testing error rate of the global model learnt by Krum increases to 0.90 when only 10\% of the clients are malicious, while the testing error rates of the global models learnt by Trim-mean and Median become larger than 0.85 when the fraction of malicious clients reaches 40\%. 

Figure \ref{adpative_attack} further shows the testing error rates of the  global models learnt by FLTrust as a function of the fraction of malicious clients under the adaptive attacks on all datasets.  Our results show that FLTrust is robust against adaptive attacks even if a large fraction of clients are malicious on all datasets.   Specifically, for MNIST-0.1 (MNIST-0.5, Fashion-MNIST, CIFAR-10, HAR, \xc{or CH-MNIST}), FLTrust under adaptive attacks with over 60\% (over 40\%, up to 60\%, up to 60\%,  up to 40\%, \xc{or over 40\%}) of malicious clients can still achieve testing error rates similar to FedAvg under no attack.   


\xc{\section{Discussion and Limitations}
\label{sec:discussion}

\myparatight{FLTrust vs. fault-tolerant computing}
Fault-tolerant computing \cite{barborak1993consensus} aims to remain functional when there are malicious clients. However, conventional fault-tolerant computing and federated learning have the following key difference: the clients communicate with each other to compute results in fault-tolerant computing \cite{barborak1993consensus}, while clients only communicate with a cloud server in federated learning. Our FLTrust leverages such unique characteristics of federated learning to bootstrap trust, i.e., the server collects a root dataset, and uses it to guide the aggregation of the local model updates.}

\myparatight{Different ways of using the root dataset} Fang et al. \cite{fang2019local} also proposed to use a root dataset (they called it validation dataset). However, we use the root dataset in a way that is different from theirs. In particular, they use the root dataset to remove potentially malicious local model updates in each iteration, while we use it to assign trust scores to clients and normalize local model updates. As shown by Fang et al. \cite{fang2019local}, their way of using the root dataset is not effective in many cases. 

\xc{\myparatight{Poisoned root dataset} Our FLTrust requires a clean root dataset. We acknowledge that FLTrust may not be robust against poisoned root dataset. The root dataset may be poisoned when it is collected from the Internet or by an insider attacker. However, since FLTrust only requires a small root dataset, a service provider can collect a clean one by itself with a small cost, e.g., asking its employees to generate and manually label a clean root dataset. }

\xc{\myparatight{Adaptive attacks and hierarchical root of trust} 
We considered an adaptive attack via extending the state-of-the-art framework of local model poisoning attacks to our FLTrust. We acknowledge that there may exist stronger local model poisoning attacks to FLTrust, which is an interesting future work to explore. Moreover, it is an interesting future work to consider a hierarchical root of trust. For instance, the root dataset may contain multiple subsets with different levels of trust. The subsets with higher trust may have a larger impact on the aggregation.}


\section{Conclusion and Future Work}
\label{sec:conclusion}
We proposed and evaluated a new federated learning method called FLTrust to achieve Byzantine robustness against malicious clients. The key difference between our FLTrust and existing federated learning methods is that the server itself collects a clean small training dataset (i.e., root dataset) to bootstrap trust in FLTrust. Our extensive evaluations on six datasets show that FLTrust with a small root dataset can achieve Byzantine robustness against a large fraction of malicious clients. In particular, FLTrust under adaptive attacks with a large fraction of malicious clients can still train global models that are as good as the global models learnt by FedAvg under no attacks. \xc{Interesting future work includes 1) designing stronger local model poisoning attacks to FLTrust and 2) considering a hierarchical root of trust.} 


\section*{Acknowledgement}
We thank the anonymous reviewers for their constructive comments. This work was supported in part by NSF grants No. 1937786, 1943226, and 2110252, an IBM Faculty Award, and a Google Faculty Research Award.

\bibliographystyle{IEEEtranS}
\bibliography{refs}
\balance{
\xc{\appendix
\subsection{Proof of Theorem~\ref{theorem_1}} \label{sec:appendix}

Before proving Theorem~\ref{theorem_1}, we first restate our FLTrust algorithm and prove some lemmas. We note that in our setting where $R_l=1$, only the combined learning rate $\alpha\cdot\beta$ influences FLTrust. Therefore, given a combined learning rate, we can always set $\beta=1$ and let $\alpha$ be the combined learning rate. In this case, the local model updates and server update are equivalent to the negative gradients of the clients and the server, respectively. With a slight abuse of notation, we use $\bm{g}_i$ and $\bm{g}_0$ to represent the gradients of the $i$th client and the server in our proof, respectively. 
We denote by $\mathcal{S}$ the set of clients whose cosine similarity $c_i$  is positive in the $t$th global iteration. 
Let $\bm{\bar{g}}_i = \frac{\left\| \bm{g}_0 \right\|}{\left\| \bm{g}_i\right\| }\cdot\bm{g}_i$ and  $\varphi_i = \frac{ReLU(c_i)}{\sum\limits_{j \in \mathcal{S}}ReLU(c_j)}=\frac{c_i}{\sum\limits_{j \in \mathcal{S}} c_j}$, where $i \in \mathcal{S}$.  
Then, we can obtain the global gradient $\bm{g}$ from Equation~(\ref{agg_local_model}) as:
\begin{align}
\label{agg_local_model_rewrite}
\bm{g} = \sum\limits_{i \in \mathcal{S}} \varphi_i \bm{\bar{g}}_i, \quad \text{s.t.} \sum\limits_{i \in \mathcal{S}}\varphi_i = 1, 0 < \varphi_i  < 1.
\end{align}

\begin{lem}
\label{lemma_1}
	For an arbitrary number of malicious clients, the distance between $\bm{g}$ and $\nabla F(\bm{w})$ is bounded as follows in each iteration:
	\begin{align}
	\left\| \bm{g} - \nabla F(\bm{w}) \right\|  
	\le 
	3 \left\| \bm{g}_0  - \nabla F(\bm{w}) \right\| +  2 \left\| \nabla F(\bm{w}) \right\|. \nonumber
	\end{align}
\end{lem}
\begin{proof}We have the following equations:
	\begin{align}
	& \left\| \bm{g} - \nabla F(\bm{w}) \right\| \\
	&= \left\|  \sum\limits_{i \in \mathcal{S}} \varphi_i \bm{\bar{g}}_i - \nabla F(\bm{w}) \right\| \\
	&= \left\| \sum\limits_{i \in \mathcal{S}} \varphi_i \bm{\bar{g}}_i - \bm{g}_0 + \bm{g}_0  -  \nabla F(\bm{w})  \right\| \\
	&\le \left\| \sum\limits_{i \in \mathcal{S}} \varphi_i \bm{\bar{g}}_i - \bm{g}_0  \right\|+  \left\| \bm{g}_0  -  \nabla F(\bm{w})  \right\| \\
	& \stackrel{(a)} \le  \left\|  \sum\limits_{i \in \mathcal{S}} \varphi_i \bm{\bar{g}}_i + \bm{g}_0  \right\| + \left\| \bm{g}_0  -  \nabla F(\bm{w})  \right\|  \\
	&\le \left\|   \sum\limits_{i \in \mathcal{S}} \varphi_i \bm{\bar{g}}_i  \right\| + \left\| \bm{g}_0   \right\| + \left\| \bm{g}_0  -  \nabla F(\bm{w})  \right\| \\
	& \le \sum\limits_{i \in \mathcal{S}}  \varphi_i \left\| \bm{\bar{g}}_i \right\| + \left\| \bm{g}_0   \right\| + \left\| \bm{g}_0  -  \nabla F(\bm{w})  \right\| \\
	& \stackrel{(b)}= \sum\limits_{i \in \mathcal{S}}  \varphi_i \left\| \bm{g}_0 \right\| + \left\| \bm{g}_0   \right\| + \left\| \bm{g}_0  -  \nabla F(\bm{w})  \right\| \\
	& \stackrel{(c)}=  2 \left\| \bm{g}_0   \right\| + \left\| \bm{g}_0  -  \nabla F(\bm{w})  \right\| \\
	&= 2 \left\| \bm{g}_0  - \nabla F(\bm{w}) + \nabla F(\bm{w})  \right\|  + \left\| \bm{g}_0  -  \nabla F(\bm{w})  \right\| \\
	&\le 2 \left\| \bm{g}_0  - \nabla F(\bm{w}) \right\| + 2 \left\| \nabla F(\bm{w}) \right\| + \left\| \bm{g}_0  -  \nabla F(\bm{w})  \right\| \\
	&= 3 \left\| \bm{g}_0  - \nabla F(\bm{w}) \right\| +  2 \left\| \nabla F(\bm{w}) \right\|, 
	\end{align}
	where $(a)$ is because $\langle \bm{\bar{g}}_i, \bm{g}_0 \rangle > 0$ for $i \in \mathcal{S}$; $(b)$ is because FLTrust normalizes the local model updates to have the same magnitude as the server model update, i.e., $\left\| \bm{\bar{g}}_i  \right\| = \left\| \bm{g}_0   \right\|$; and $(c)$ is because $\sum\limits_{i \in \mathcal{S}}\varphi_i = 1$.
\end{proof}

\begin{lem}
	\label{lemma_2}
	Assume Assumption~\ref{assumption_1} holds. If we set the learning rate as $\alpha = \mu/(2L^2)$, then we have the following in any global iteration $t\ge 1$:
	\begin{align}
	&\left\| \bm{w}^{t-1} - \bm{w} ^* - \alpha \nabla F(\bm{w}^{t-1}) \right\| \nonumber \\ 
	&\le \sqrt {1 - {\mu^2}/(4L^2)} \left\| \bm{w}^ {t-1}  - \bm{w} ^* \right\|. \nonumber
	\end{align}
\end{lem}
\begin{proof}
	Since $\nabla F({\bm{w} ^*}) = 0$, we have the following:
	\begin{align} 
	& {\left\| \bm{w}^{t-1} - \bm{w} ^* - \alpha \nabla F(\bm{w}^{t-1}) \right\|^2} \\
	&= {\left\| \bm{w}^{t-1} - \bm{w} ^* - \alpha \left( {\nabla F(\bm{w}^{t-1}) - \nabla F({\bm{w} ^*})} \right) \right\|^2} \\
	&= {\left\| \bm{w}^{t-1} - \bm{w} ^* \right\|^2} + \alpha ^2{\left\| {\nabla F(\bm{w}^{t-1}) - \nabla F({\bm{w} ^*})} \right\|^2}  \nonumber \\
	& \quad - 2\alpha \left\langle {\bm{w}^{t-1}  - \bm{w} ^*,\nabla F(\bm{w}^{t-1}) - \nabla F({\bm{w}^*})} \right\rangle.
	\label{w_to_prove}
	\end{align}	
	
	By Assumption~\ref{assumption_1}, we have:
	\begin{align}
	& \left\| {\nabla F(\bm{w}^{t-1}) - \nabla F({\bm{w}^*})} \right\| \le L \left\| {\bm{w}^{t-1} - \bm{w} ^*} \right\|, \label{L_strongly_convex_equ_L}\\
	& F({\bm{w} ^*}) + \left\langle {\nabla F({\bm{w} ^*}), \bm{w}^{t-1} - \bm{w} ^*} \right\rangle  \le F(\bm{w}^{t-1}) \nonumber\\
	& \qquad -\frac{\mu}{2}{\left\| {\bm{w}^{t-1}  - \bm{w} ^*} \right\|^2}, \label{L_strongly_convex_equ} \\
	& F(\bm{w}^{t-1}) + \left\langle {\nabla F(\bm{w}^{t-1}), \bm{w} ^*  - \bm{w}^{t-1}} \right\rangle  \le F({\bm{w} ^*}). \label{convex_equ}
	\end{align}
	
	Summing up inequalities~(\ref{L_strongly_convex_equ}) and~(\ref{convex_equ}), we have:
	\begin{align}
	&\left\langle {\bm{w} ^* - \bm{w}^{t-1},\nabla F(\bm{w}^{t-1}) - \nabla F(\bm{w} ^*)} \right\rangle  \nonumber\\
	& \le - \frac{\mu}{2}{\left\| {\bm{w}^{t-1}  - \bm{w} ^*} \right\|^2}.
	\label{L_strongly_convex_equ_with_convex_equ}
	\end{align}
	
	Substituting inequalities~(\ref{L_strongly_convex_equ_L}) and (\ref{L_strongly_convex_equ_with_convex_equ}) into~(\ref{w_to_prove}), we have:
	\begin{align}
	&{\left\| \bm{w}^{t-1} - \bm{w} ^* - \alpha \nabla F(\bm{w}^{t-1}) \right\|^2} \nonumber \\
	&\le \left( 1 + \alpha^2 L^2 - \alpha \mu \right){\left\| \bm{w}^{t-1}  - \bm{w} ^* \right\|^2}.
	\end{align}
	
	By choosing $\alpha = \mu/(2L^2)$, we have:
	\begin{align}
	&{\left\| \bm{w}^{t-1} - \bm{w} ^* - \alpha \nabla F(\bm{w}^{t-1}) \right\|^2}  \nonumber \\
	& \le \left( 1 - \mu^2/(4L^2)  \right){\left\| \bm{w}^{t-1}  - \bm{w} ^* \right\|^2},
	\end{align}
	which concludes the proof.
\end{proof}

\begin{lem}
	\label{lemma_3}
	Suppose Assumption~\ref{assumption_2} holds. For any $\delta  \in (0,1)$ and any $\bm{w} \in \Theta$, we let ${\Delta _1} = \sqrt 2 {\sigma_1 } \sqrt {(d\log 6 + \log (3/ \delta))/ {\left| {D_0} \right|}}$ and ${\Delta _3} = \sqrt 2 {\sigma_2} \sqrt {(d\log 6 + \log (3 / \delta))/ {\left| {D_0} \right|}}$. 
	If ${\Delta _1} \le {\sigma_1^2} / {\gamma_1}$ and ${\Delta _3} \le {\sigma_2^2} / {\gamma_2}$, then we have:
	\begin{align}
	& \text{Pr} \left\{ {\left\| {\frac{1}{\left| D_0 \right|}\sum\limits_{X_i \in {D_0}} {\nabla f({X_i},{\bm{w}^*})}  - \nabla F({\bm{w}^*})} \right\| \ge 2{\Delta _1}} \right\} \le \frac{\delta }{3}, \nonumber
	\\
	& \text{Pr} \left\{ \left\| {\frac{1}{{\left| {{D_0}} \right|}}\sum\limits_{X_i \in {D_0}} {\nabla h({X_i},{\bm{w}})}  - {\mathbb{E}\left[ {h(X,\bm{w})} \right]}} \right\| \right. \nonumber \\
	& \left. \qquad 
	\ge 2{\Delta _3} {\left\| {\bm{w} - {\bm{w}^*}} \right\|}  \vphantom{\frac{1}{\left| {D_0} \right|} \sum\limits_{X_i \in {D_0}}}  \right\} 
	\le \frac{\delta }{3}.  \nonumber
	\end{align}
\end{lem}	
\begin{proof}
	We  prove the first inequality of Lemma~\ref{lemma_3}. The proof of the second inequality is similar, and we omit it  for brevity.
	Let $\bm{V} = \{ {\bm{v}_{1,}}, \cdots, {\bm{v}_{N_{\frac{1}{2}}}\}}$ be an $\frac{1}{2}$-cover of the unit sphere $\bm{B}$. It is shown in~\cite{ChenPOMACS17,vershynin2010introduction} that we have $\log {N}_{\frac{1}{2}} \le d\log 6$ and the following:
	\begin{align}
	\label{vershynin2010introduction_equ}
	& \left\| {\frac{1}{{\left| {{D_0}} \right|}}\sum\limits_{X_i \in {D_0}} {\nabla f({X_i},{\bm{w}^*})}  - \nabla F({\bm{w}^*})} \right\| \le \nonumber \\
	&  2\mathop {\sup }\limits_{\bm{v} \in \bm{V}} \left\{ {\left\langle {\frac{1}{{\left| {{D_0}} \right|}}\sum\limits_{X_i \in {D_0}} {\nabla f({X_i},{\bm{w}^*})}  - \nabla F({\bm{w}^*}),\bm{v}} \right\rangle } \right\}. 
	\end{align}
	
	According to the concentration inequalities for sub-exponential random variables~\cite{wainwright2019high}, when Assumption~\ref{assumption_2} and  condition $\Delta _1 \le \sigma_1 ^2 / \gamma_1$ hold, we have: 
	\begin{align}
	&\textit{Pr} \left\{ {\left\langle {\frac{1}{{\left| {{D_0}} \right|}}\sum\limits_{X_i \in {D_0}} {\nabla f({X_i},{\bm{w}^*})}  - \nabla   F({\bm{w}^*}),\bm{v}} \right\rangle  \ge {\Delta _1}} \right\} \nonumber \\
	& \le \exp \left( { - \left| {{D_0}} \right|\Delta _1^2} \mathord{\left/\right.	\kern-\nulldelimiterspace} (2\sigma_1 ^2) \right).
	\end{align}
	
	Taking the union bound over all vectors in $\bm{V}$ and combining it with inequality (\ref{vershynin2010introduction_equ}), we have:
	\begin{align}
	\label{last_equ_lemma3}
	& \textit{Pr} \left\{ {\left\| {\frac{1}{{\left| {{D_0}} \right|}}\sum\limits_{X_i \in {D_0}} {\nabla f({X_i},{\bm{w}^*})}  - \nabla F({\bm{w}^*})} \right\| \ge 2{\Delta _1}} \right\}  \nonumber \\
	& \le \exp \left( { - \left| {{D_0}} \right|\Delta _1^2} / (2\sigma_1 ^2) +d\log 6 \right).
	\end{align}
	
	We conclude the proof by letting ${\Delta _1} = \sqrt 2 {\sigma_1 } \sqrt {(d\log 6 + \log (3/ \delta))/ {\left| {{D_0}} \right|}}$ in (\ref{last_equ_lemma3}). 
\end{proof}

\begin{lem}
	\label{lemma_4}
	Suppose Assumptions~\ref{assumption_1}-\ref{assumption_3} hold and $\Theta  \subset \left\{ {\bm{w}:\left\| \bm{w} - {\bm{w}^*} \right\| \le r\sqrt d } \right\}$ holds for some positive parameter $r$.  
	Then, for any  $\delta  \in (0,1)$, if $\Delta _1 \le \sigma_1 ^2 / \gamma_1$ and $\Delta_2 \le \sigma_2 ^2 / \gamma_2$, we have the following for any $\bm{w} \in  \Theta$:
	\begin{align}
	\text{Pr} \left\{ {\left\| \bm{g}_0 - \nabla F(\bm{w}) \right\| \le 8\Delta_2 {\left\| {\bm{w} - \bm{w}^*} \right\|} + 4\Delta_1} \right\} \ge 1 - \delta,  \nonumber
	\end{align}
	where $\Delta_2 = {\sigma_2}\sqrt {\frac{2}{{\left| D_0 \right|}}} \sqrt {K_1 + K_2}$, $K_1 = d\log \frac{18L_2}{\sigma_2}$, $K_2 = \frac{1}{2}d\log \frac{\left| {D_0} \right|}d + \log \left( \frac{{6\sigma_2^2r\sqrt {\left| D_0 \right|} }}{\gamma_2{\sigma_1}\delta } \right)$,
	${L_2} = \max \left\{ {L,{L_1}} \right\}$, and $\left| D_0 \right|$ is the size of the root dataset.
\end{lem}
\begin{proof}
	Our proof is mainly based on the $\varepsilon$-net argument~\cite{vershynin2010introduction} and~\cite{ChenPOMACS17}.
	We let $\tau = \frac{\gamma_2\sigma_1}{2\sigma_2^2}\sqrt{\frac{d}{\left| D_0 \right|}}$ and $\ell^*$ be an integer that satisfies $\ell^* = {\left\lceil r\sqrt d/\tau \right\rceil}$.
	For any integer $1 \le \ell \le \ell^*$, we define 
	$
	{\Theta_\ell } = \left\{ {\bm{w}:\left\| {\bm{w} - {\bm{w}^*}} \right\| \le \tau\ell } \right\}.
	$
	Given an integer $\ell$, we let ${\bm{w}_1},\cdots,{\bm{w}_{{N}_{\varepsilon_\ell }}}$ be an $\varepsilon_\ell$-cover of ${\Theta _\ell }$, where $\varepsilon_\ell = \frac{{\sigma_2}\tau \ell }{L_2}\sqrt {\frac{d}{{\left| {D_0} \right|}}} $ and ${L_2} = \max \left\{ {L,{L_1}} \right\}$.
	 From~\cite{vershynin2010introduction}, we know that $\log {{N}_{{\varepsilon_\ell }}} \le d\log \left( {\frac{{3\tau\ell }}{\varepsilon_\ell }} \right)$. 
	 For any $\bm{w} \in {\Theta_\ell}$, there exists a $j_\ell$ ($1 \le j_\ell \le {{N}_{\varepsilon_\ell}}$) such that:
	 \begin{align}
	 \label{w_j_equ}
	 \left\| {\bm{w} - \bm{w}_{j_\ell} } \right\| \le {\varepsilon_\ell }.
	 \end{align}
	According to the triangle inequality, we have:
	\begin{align}
	& \left\| {\frac{1}{{\left| {D_0} \right|}}\sum\limits_{X_i \in {D_0}} {\nabla f({X_i},\bm{w}) - \nabla F(\bm{w})} } \right\| 
	 \le 
	 \left\| {\nabla F(\bm{w}) - \nabla F({\bm{w}_{j_\ell} })} \right\|  \nonumber \\
	&  \quad + \left\| {\frac{1}{{\left| {D_0} \right|}}\sum\limits_{X_i \in D_0} {\left( {\nabla f({X_i},\bm{w}) - \nabla f({X_i},\bm{w}_{j_\ell})} \right)} } \right\| 	& \nonumber \\
	& \quad + \left\| {\frac{1}{{\left| {D_0} \right|}}\sum\limits_{X_i \in {D_0}} {\nabla f({X_i},{\bm{w}_{j_\ell}}) - \nabla F(\bm{w}_{j_\ell})} } \right\| .
	\end{align}	

	According to Assumption~\ref{assumption_1} and inequality~(\ref{w_j_equ}), we have:
	\begin{align}
	\label{T_1_equ}
	\left\| {\nabla F(\bm{w}) - \nabla F({\bm{w}_{j_\ell}})} \right\| \le L\left\| {\bm{w} - \bm{w}_{j_\ell}} \right\| \le L{\varepsilon_\ell }
	\end{align}
	
	Next, we define an event $\mathcal{E}_1$ as follows:
	\begin{align}
	\mathcal{E}_1 = \left\{ {\mathop {\sup }\limits_{\bm{w},\widehat{\bm{w}} \in \Theta :\bm{w} \ne \widehat{\bm{w}}} \frac{{\left\| {\nabla f(X,\bm{w}) - \nabla f(X,\widehat{\bm{w}})} \right\|}}{{\left\| {\bm{w} -\widehat{\bm{w}}} \right\|}} \le {L_1}} \right\} \nonumber.
	\end{align}
	
	According to Assumption~\ref{assumption_2}, we have $\textit{Pr} \left\{ \mathcal{E}_1 \right\} \ge 1 - \frac{\delta }{3}$. Moreover, we have the following:
	\begin{align}
	\label{T_2_equ}
	& \mathop {\sup }\limits_{\bm{w} \in \Theta } \left\| {\frac{1}{{\left| {{D_0}} \right|}}\sum\limits_{X_i \in {D_0}} {\left( {\nabla f({X_i},\bm{w}) - \nabla f({X_i},{\bm{w}_{j_\ell}})} \right)} } \right\| \nonumber \\
	& \le  L_1 \left\| {\bm{w} - \bm{w}_{j_\ell}} \right\| \le  L_1{\varepsilon_\ell}.
	\end{align}
	
	According to the triangle inequality, we have:
	\begin{equation}
	\begin{aligned}[b]
	& \left\| {\frac{1}{{\left| {D_0} \right|}}\sum\limits_{X_i \in {D_0}} {\nabla f({X_i},{\bm{w}_{j_\ell}}) - \nabla F({\bm{w}_{j_\ell}})} } \right\|   \\
	& \le \left\| {\frac{1}{{\left| {D_0} \right|}}\sum\limits_{X_i \in {D_0}} {\nabla f({X_i},{\bm{w}^*}) - \nabla F({\bm{w}^*})} } \right\|   \\
	&  \quad +  \left\| \frac{1}{{\left| {D_0} \right|}}\sum\limits_{X_i \in {D_0}} {\left( {\nabla f({X_i},{\bm{w}_{j_\ell}}) - \nabla f({X_i},\bm{w}^*)} \right)}   
	\right. \\
	& \qquad \qquad \left. - \left({ \nabla F({\bm{w}_{j_\ell}})- \nabla F({\bm{w}^*})} \right) 
	\vphantom{\frac{1}{{\left| {D_0} \right|}}\sum\limits_{X_i \in {D_0}} {\left( {\nabla f({X_i},{\bm{w}_{j_\ell}}) - \nabla f({X_i},\bm{w}^*)} \right)} }
	\right\|   \\
	& \stackrel{(a)} \le 	\left\| {\frac{1}{{\left| {{D_0}} \right|}}\sum\limits_{X_i \in {D_0}} {\nabla f({X_i},{\bm{w}^*}) - \nabla F({\bm{w}^*})} } \right\| 	 \\
	& \quad + \left\| {\frac{1}{{\left| {{D_0}} \right|}}\sum\limits_{X_i \in {D_0}} {h({X_i},{\bm{w}_{j_\ell}}) - \mathbb{E} \left[ {h\left( {X,{\bm{w}_{j_\ell}}} \right)} \right]} } \right\| ,
	\end{aligned}
	\end{equation}
	where $(a)$ is due to $\mathbb{E}\left[ {h\left( {X,\bm{w}} \right)} \right] = \nabla F(\bm{w}) - \nabla F({\bm{w}^*})$.

	We also define events $\mathcal{E}_2$ and $\mathcal{E}_3(\ell)$ as:
	\begin{equation}
	\begin{aligned}[b]
	& \mathcal{E}_2 =	\left\{  {\left\| {\frac{1}{\left| {{D_0}} \right|}  \sum\limits_{X_i \in {D_0}} {\nabla f({X_i},{\bm{w}^*})}  - \nabla F({\bm{w}^*})} \right\| \le 2{\Delta _1}}  \right\}, \\
	& \mathcal{E}_3(\ell) = \left\{ \mathop {\sup }\limits_{1 \le {j } \le {N}_{\varepsilon}}  {\left\| \frac{1}{{\left| {D_0} \right|}}\sum\limits_{X_i \in {D_0}} {h({X_i},\bm{w}_j) - \mathbb{E} \left[ {h\left( {X,\bm{w}_j} \right)} \right]}  \right\|  }  
	\right.  \\
	& \qquad \qquad  \left.	\le 2 \Delta _2 \tau \ell \vphantom { \mathop {\sup }\limits_{1 \le {j } \le {\left| D \right|}_{\varepsilon}} }  \right\}. \nonumber 
	\end{aligned}
	\end{equation}

	 According to Lemma~\ref{lemma_3} and~\cite{ChenPOMACS17}, $\Delta _1 \le \sigma_1 ^2 / \gamma_1$, and $\Delta _2 \le \sigma_2 ^2 / \gamma_2$, we have $\textit{Pr} \left\{ \mathcal{E}_2 \right\} \ge 1 - \frac{\delta }{3}$ and $\textit{Pr} \left\{ \mathcal{E}_3(\ell) \right\} \ge 1 - \frac{\delta }{3\ell^*}$.
	Therefore, on event $  \mathcal{E}_1 \cap  \mathcal{E}_2 \cap   \mathcal{E}_3(\ell)$, we have:
	\begin{align}
	&\mathop {\sup }\limits_{\bm{w} \in {\Theta _\ell }} \left\| {\frac{1}{{\left| {D_0} \right|}}\sum\limits_{X_i \in {D_0}} {\nabla f({X_i},\bm{w}) - \nabla F(\bm{w})} } \right\|   \nonumber\\
	& \le L\varepsilon_\ell + L_1 \varepsilon_\ell  + 2{\Delta _1} + 2{\Delta _2} \tau \ell,   \\ 
	& \stackrel{(a)} \le 2L_2\varepsilon_\ell + 2{\Delta _1} + 2{\Delta _2} \tau \ell
	 \stackrel{(b)} \le 4{\Delta _2}\tau\ell + 2{\Delta _1},
	\end{align}
	where $(a)$ holds because $(L + {L_1})\le 2{L_2}$ and $(b)$ is due to $\Delta _2 \ge \sigma_2\sqrt{d/{\left| {D_0} \right|}}$.
	
	Thus, according to the union bound, we have probability at least $1 - \delta$ that event  $  \mathcal{E}_1 \cap  \mathcal{E}_2 \cap (\cap_{\ell=1}^{ \ell^* } \mathcal{E}_3(\ell))$ holds.
	On event $  \mathcal{E}_1 \cap  \mathcal{E}_2 \cap (\cap_{\ell=1}^{ \ell^* } \mathcal{E}_3(\ell))$, for any $\bm{w} \in \Theta _{\ell^*}$, there exists an $1 \le \ell \le \ell^*$ such that $(\ell -1)\tau < \left\| {\bm{w} - \bm{w}^*} \right\| \le \ell \tau$ holds.
	If $\ell =1$, then we have:
	\begin{align}
	\label{ell_1_bound}
	& \left\| \frac{1}{{\left| {D_0} \right|}}\sum\limits_{X_i \in {D_0}} \nabla f({X_i},\bm{w}) - \nabla F(\bm{w}) \right\|   \nonumber\\
	& \le 4{\Delta _2} \tau + 2{\Delta _1}
	 \stackrel{(a)} \le 4 \Delta _1,
	\end{align}
	where $(a)$ holds because $\Delta _2 \le \sigma_2 ^2 / \gamma_2$ and $\Delta _1 \ge \sigma_1\sqrt{d/{\left| {D_0} \right|}}$.
	If $\ell \ge 2$, then we have $2(\ell -1) \ge \ell$ and the following:
	\begin{align}
	\label{ell_2_bound}
	& \left\| \frac{1}{{\left| {D_0} \right|}}\sum\limits_{X_i \in {D_0}} \nabla f({X_i},\bm{w}) - \nabla F(\bm{w}) \right\|   \nonumber\\
	& \le 8{\Delta _2} \left\| {\bm{w} - \bm{w}^*} \right\| + 2{\Delta _1}.
	\end{align}
	
	Combining inequalities~(\ref{ell_1_bound}) and~(\ref{ell_2_bound}), we have:
	\begin{align}
	&\mathop {\sup }\limits_{\bm{w} \in {\Theta _{\ell^*} }} \left\| \frac{1}{{\left| {D_0} \right|}}\sum\limits_{X_i \in {D_0}} \nabla f({X_i},\bm{w}) - \nabla F(\bm{w}) \right\|   \nonumber\\
	& \le 8{\Delta _2} \left\| {\bm{w} - \bm{w}^*} \right\| + 4{\Delta _1}.
	\end{align}

	We conclude the proof since $\Theta \subset \Theta _{\ell^*} $ and  $\bm{g}_0= \frac{1}{{\left| {D_0} \right|}} \sum\limits_{X_i \in {D_0}} \nabla f({X_i},\bm{w})$.   
\end{proof}

 \myparatight{Proof of Theorem~\ref{theorem_1}} With the lemmas above, we can prove Theorem~\ref{theorem_1} next. We have the following equations for the $t$th global iteration: 
\begin{align}
	&\left\| \bm{w}^t - \bm{w}^* \right\| \nonumber\\
	&= \left\| \bm{w}^{t-1} - \alpha \bm{g}^{t-1} - \bm{w}^* \right\|  \nonumber\\
	&= \left\| \bm{w}^{t-1} - \alpha \nabla F(\bm{w}^{t-1}) - \bm{w}^* +  \alpha \nabla F(\bm{w}^{t-1}) - \alpha \bm{g}^{t-1} \right\| \nonumber\\
	&\le \left\| \bm{w}^{t-1} - \alpha \nabla F(\bm{w}^{t-1}) - \bm{w}^* \right\| + \alpha \left\| \bm{g}^{t-1} - \nabla F(\bm{w}^{t-1}) \right\| \nonumber\\
	& \stackrel{(a)}\le \left\| \bm{w}^{t-1} - \alpha \nabla F(\bm{w}^{t-1}) - \bm{w}^* \right\| +  3\alpha \left\| \bm{g}_0^{t-1} - \nabla F(\bm{w}^{t-1}) \right\| \nonumber\\
	& \quad +  2\alpha \left\| \nabla F(\bm{w}^{t-1}) \right\|  \nonumber\\
	& \stackrel{(b)}= \underbrace{ \left\| \bm{w}^{t-1} - \alpha \nabla F(\bm{w}^{t-1}) - \bm{w}^* \right\| }_{A_1} +  3\alpha \underbrace{\left\| \bm{g}_0^{t-1} - \nabla F(\bm{w}^{t-1}) \right\|}_{A_2} \nonumber\\
	& \quad +  2\alpha  \underbrace{\left\| \nabla F(\bm{w}^{t-1}) - \nabla F(\bm{w}^*) \right\| }_{A_3} \nonumber\\
	& \stackrel{(c)} \le \sqrt {1 - {\mu^2}/(4L^2)} \left\| \bm{w}^{t-1} - \bm{w}^* \right\| + 2\alpha L \left\|  \bm{w}^{t-1} - \bm{w}^* \right\| \nonumber\\
	& \quad + 3\alpha \left( 8\Delta_2 \left\|  \bm{w}^{t-1} - \bm{w}^*  \right\| + 4\Delta_1 \right) \nonumber\\
	&= \left( \sqrt {1 - {\mu^2}/(4L^2)} + 24\alpha\Delta_2 + 2\alpha L  \right)  \left\|  \bm{w}^{t-1} - \bm{w}^*  \right\| \nonumber\\
	& \quad  + 12\alpha\Delta_1,
	\end{align}
where $(a)$ is obtained based on Lemma~\ref{lemma_1}; $(b)$ is due to $\nabla F(\bm{w}^*)=0$; and $(c)$ is obtained by plugging Lemma~\ref{lemma_2}, Lemma~\ref{lemma_4}, and Assumption~\ref{assumption_1} into $A_1$, $A_2$, and $A_3$, respectively. By recursively applying the inequality for each global iteration, we have:
\begin{align}
\left\| \bm{w}^t - \bm{w}^* \right\| 
\le
\left( 1- \rho \right)^t \left\| \bm{w}^0 - \bm{w}^* \right\| +  12\alpha\Delta_1/ \rho,
\end{align}
where $\rho = 1-  \left( \sqrt {1 - {\mu^2}/(4L^2)} + 24\alpha\Delta_2 + 2\alpha L  \right)$. Thus, we conclude the proof.
}
}


\end{document}